\documentclass[12pt]{article}
\pdfoutput=1
\usepackage[T1]{fontenc}
\usepackage[final]{microtype}

\usepackage[margin=1.05in]{geometry}

\usepackage{physics}
\usepackage{mathtools}

\usepackage{amsthm}
\usepackage{amsbsy}
\usepackage{amsmath,amsfonts,graphicx,mathdots,wrapfig}
\usepackage{hyperref}

\usepackage{tikz}

\usepackage{tikz-cd}
\usetikzlibrary{arrows.meta,decorations.pathreplacing,positioning}

\usepackage{wrapfig}
\usepackage[style = alphabetic]{biblatex}
\newtheorem{theorem}{Theorem}[section]
\newtheorem{proposition}[theorem]{Proposition}

\newtheorem{lemma}[theorem]{Lemma}

\DeclareMathOperator{\C}{\mathbb{C}}

\newcommand{\OC}{\mathcal{O}}
\newcommand{\bb}[1]{\mathbb{#1}}

\newtheorem*{theorem*}{Theorem}
\newtheorem*{proposition*}{Proposition}
\theoremstyle{plain}
\newcommand{\Z}{\mathbb{Z}}

\usepackage{xcolor}
\usepackage{float}

\usepackage[all]{xy}
\usepackage{amssymb}
\usepackage{amsmath}
\usepackage{graphicx}
\usepackage{caption}
\usepackage{latexsym}
\usepackage{mathrsfs}
\usepackage{verbatim}
\usepackage{multicol}
\usepackage{url}
\usepackage{hyperref}
\usepackage{mathtools}
\usepackage{ytableau}

\usepackage{longtable,tabularx}
\usepackage{subcaption, comment}

\theoremstyle{definition}
\newtheorem{definition}[theorem]{Definition}

\newtheorem{example}[theorem]{Example}

\newcommand{\xb}{\boldsymbol{x}}
\newcommand{\db}{\boldsymbol{d}}
\newcommand{\ub}{\boldsymbol{u}}

\title{Quantum Difference Equations for Grassmannians}
\author{Xingyu Cheng, Reese Lance, Nikhil Nagabandi, Andrey Smirnov}
\date{}

\begin{document}

\maketitle

\begin{abstract}
    We consider quantum difference equation (QDE) for equivariant quantum K-theory of the Grassmannian. In this paper we obtain a solution to the QDE and use the solution to asymptotically derive the Bethe ansatz equations. In the limit, we obtain similar results for the cohomological analogue. For both cases, we describe the nonequivariant solutions as well. As an application, we identify the quantum K-theory ring of $\mathrm{Gr}(k,n)$ with a quantum 5 vertex XXZ integrable spin chain.
\end{abstract}
\tableofcontents

\section{Introduction}

The intersection of quantum geometry and the study of quantum integrable systems is a rich area with active research from both mathematicians and physicists. Originally, a correspondence between quantum integrable systems and quantum K-theory was observed in a physical context by Nekrasov and Shatashvili \cite{nekrasov2009}, \cite{nekrasovq2009}, \cite{nekrasov2010}. In a mathematical context, a significant amount of work has been done to study this connection for the quantum K-theory of symplectic algebraic varieties \cite{givental2001toda}, \cite{maulikokounkov}, \cite{Okounkov2022-tf}, mainly through a representation theory lens. 
In particular, the quantum K-theory of the cotangent bundle over flag varieties has been explored \cite{rimanyitarasovvarchenko}. One of the simplest cases, the quantum K-theory of the cotangent bundle over the Grassmannian was recently studied in \cite{pushkarsmirnovzeitlin}. The quantum K-theory and quantum cohomology rings of the bundle relate to the $\mathrm{XXZ}$ and $\mathrm{XXX}$ spin chains, respectively. \cite{pushkarsmirnovzeitlin} expanded on this connection, including the description of how the operator of multiplication by the weighted exterior algebra of the tautological bundle in the quantum K-theory coincides with the Baxter Q-operator for the XXZ spin chain. This operator appears in the quantum difference equation they consider and give solutions to. In this paper, we establish analogous results for $Gr(k,n)$, focused on the quantum cohomology and K theory rings. We rely on the numerous geometric results found for these spaces (e.g. the Pieri rule) \cite{givental2000wdvvequationquantumktheory}, \cite{buch2001quantumcohomologygrassmannians}, \cite{mihalcea2004equivariantquantumschubertcalculus}, \cite{Buch_2011}.

Mihalcea \cite{mihalcealectures} studies a relation between integrable systems and the quantum K-theory of flag varieties. Our paper focuses on the (torus) equivariant quantum K-theory of the Grassmannian $\mathrm{Gr}(k,n)$ of dimension $k$ subspaces inside $\C^n$. The main result of this paper is the solution to a quantum difference equation (QDE) for this space. With quantum parameter $z$ the QDE is 
\begin{equation}\label{k-th-eqn-intro}
\Psi(qz) \OC(-1) = M(z) \Psi(z).
\end{equation}
Here $\OC(-1)$ denotes tensoring by the determinant of the tautological subbundle of $\mathrm{Gr}(k,n)$ acting as in regular K-theory, $M(z)$ is the operator of quantum multiplication by $\OC(-1)$ in K-theory, and $q$ is a shift parameter. This operator $M(z)$ commutes with the Hamiltonian from quantum XXZ spin chains, providing a link between the geometry of $\mathrm{Gr}(k,n)$ and integrable systems. To find the Hamiltonian, we begin by constructing the transfer matrix $\mathrm{T}(x)$ from the $R$ matrix. The $R$ matrix (\ref{ktheoryrmatrix}) is associated with the $\mathrm{XXZ}$ spin chain with $n$ particles. Define the monodromy matrix $ \mathrm{L}(x)=R_{a1} (x/u_1)\cdots R_{an} (x/u_n)$, where $a$ denotes the auxiliary space. Then writing $\mathrm{L}(x)$ as
$$\mathrm{L}(x)= \begin{pmatrix}
    A(x) & B(x) \\ 
    C(x) & D(x)
\end{pmatrix},$$
we perform a ``twisted" trace to obtain the transfer matrix, defined as $\mathrm{T}(x) = zA(x)+D(x)$. Expanding $\ln (\mathrm{T}(x))$ asymptotically yields a family of commuting Hamiltonians, which coincide with those commuting with $M(z)$. As a corollary, coefficients of $\mathrm T (x)$ are operators of quantum multiplication by K-theory classes. This identifies corresponding Bethe algebra with the K-theory ring.

To construct the solutions of the QDE, we take inspiration from \cite{aganagic2017quasimapcountsbetheeigenfunctions}, to construct the solutions from partition functions associated to the space. These partition functions (more in Section \ref{quantumktheorysection}) correspond to the double Grothendieck polynomials $G_\lambda (\xb; \ub )$. With these we have the following solution to the QDE:
\begin{theorem*}[Theorem \ref{mainthm}]
For each partition $\lambda$ in the $k \times (n-k)$ rectangle,  the following class in $\mathrm{K}_T(\mathrm{Gr}(k,n))[[z]]$ solves QDE (\ref{k-th-eqn-intro}):
\begin{equation} 
\label{eqn: mainthm}
     \Psi_\lambda(x_1,\dots,x_k,z)
    =\sum_{(d_1, \ldots, d_n) \geq 0} \Phi(x_1q^{d_1},\dots, x_kq^{d_k}, \ub )\, G_{\lambda} \left( x_1q^{d_1},\dots, x_kq^{d_k}; \ub\right)\, z^{|\boldsymbol{d}|}, 
\end{equation}
where $\Phi$ is defined in Theorem (\ref{mainthm}). 
\end{theorem*}

In principle, this same solution can be found by taking the limit $h\to \infty$ of the QDE solution of \cite{rimanyitarasovvarchenko}, \cite{aganagic2017quasimapcountsbetheeigenfunctions}, \cite{pushkarsmirnovzeitlin}. However, our solution is derived using only combinations of Grothendieck polynomials, which is a novel approach that is also more basic and direct. The case of the Grassmannian is more symmetric and richer in properties than the case of the cotangent bundle of the Grassmannian. So a systematic study of this case from the first principles is needed. 

For the proof of our solution we use the Pieri rule, and a combinatorial relation among double Grothendieck polynomials to write the difference of the left and right sides of Equation \ref{eqn: mainthm} into telescoping sums.

From the solution $\Psi(z)$, the Bethe Ansatz equations and Bethe vectors can be obtained via the saddle point method. These provide another link to physics, and the roots of the Bethe equations lead to eigenvalues and eigenvectors of $M(z)$.

We also explore the analogous quantum differential equation (qde) in quantum equivariant cohomology given by
\begin{equation}
    \epsilon z\frac{\partial \psi(z)}{\partial z} = M^{C } (z) \psi(z),
\end{equation}
where $\epsilon \in \C^\times$ is a parameter and $M^C (z)$ denotes the operation of quantum multiplication by $-c_1 (\mathcal O(-1))$ in $\mathrm{QH}_{T}^* (\textrm{Gr}(k,n))$. In this scenario we will give solutions for $\textrm{Gr}(1,n) = \mathbb P^{n-1}$, and then use an application of geometric Satake correspondence due to \cite{cotti2024satakecorrespondenceequivariantquantum} to derive the solutions for $\textrm{Gr}(k,n)$. 

Our QDE extends to further directions of both physics and geometry. Our solution can be viewed a generalization of the $J$-function. Ueda and Yoshida describe the $J$-function and connections to 3d mirror symmetry in \cite{ueda3d}. Also from our QDE, one can derive the solutions to the qKZ equations. From this, one can derive similar geometric results found involving an affine Weyl group action on the quantum K-theory ring \cite{gorbounov2025quantumktheorygrassmanniansyangbaxter}, which extends the Seidel action \cite{LLSYseidel}.

\subsection{Layout of the paper}


In Section \ref{background}, we provide background and fix notation. Then we describe the details that will be used for the cohomology case, including the representatives of basis elements and the Pieri rule. Then the K theory case is described similarly, but with the addition of a combinatorial rule for the double Grothendieck polynomials.

We begin with the qde for cohomology  $\mathrm{QH}^*_T(\mathrm{Gr}(k,n))$ in Section \ref{quantumcohomologysection}. We define the solutions to our qde in terms of the Gamma function and the partition functions. To find the partition functions, we first find the $R$ matrix for the space geometrically. The partition functions coincide with the factorial Schur polynomials. For the solution to the qde, we first find solutions for $\mathrm{Gr}(1,n)$, then using an application of geometric Satake, the $\mathrm{Gr}(k,n)$ solution is obtained as certain $k$-wedge products of solutions for $\mathrm{Gr}(k,n)$. We finish the section describing how to asymptotically obtain the Bethe ansatz equations and what the solution looks like in the nonequivariant case.

In Section \ref{quantumktheorysection}, we describe the QDE for the K-theory case $\mathrm{QK}_T(\mathrm{Gr}(k,n))$, and the representatives of each piece of the equation. The  rest of the section is organized similarly to Section \ref{quantumcohomologysection}. We geometrically find the $R$-matrix and find the partition functions used for the solution for this space. We then describe the explicit form of the solution $\Psi(z)$ and then prove this is a solution the QDE. Then we asymptotically obtain the Bethe ansatz equations and describe the nonequivariant case.

\subsection{Acknowledgments}

All 4 authors were supported by NSF-DMS 2401380.

\section{Notations and background} \label{background}
The main space of study is $\textrm{Gr}(k,n)$, the Grassmannian of $k$-planes in $\C^n$. 
A partition $\lambda =(\lambda_1,\dots, \lambda_k)$ is a decreasing sequence of nonnegative integers: $\lambda_1 \geq \lambda_2\geq \cdots \geq \lambda_k \geq 0$. Each partition $\lambda=(\lambda_1,\dots,\lambda_k)$ corresponds to a Young diagram contained in a $k \times (n-k)$ rectangle. We use the terms partition and Young diagram interchangeably. Each of these partitions has a corresponding $k$-$subset$ of $\{1,\dots,n\}$, often referred to as a \emph{frame}. 

The $k$-subset corresponding to $\lambda$ is obtained by the following procedure: begin drawing a path on the $k\times(n-k)$ rectangle. Start at the bottom left corner of the rectangle and travel east along the bottom border of the Young diagram until the rightmost point. Then move vertically until there is a top edge in the diagram. Repeat this procedure and stop when the top right corner of the rectangle is reached (if at the ceiling of the rectangle, move east). The path should have $n$ steps. Label each step of the path with $1,\dots,n$, starting at the bottom left corner. The labels on the vertical steps form the $k$-subset. We will denote $k$-subsets with $r$.

For example, consider the case $k=5,n=9$ and $\lambda = (4,2,1,1,0)$. Graphically, the conversion is drawn in Figure \ref{fig: convert partition to subset}. The Young diagram corresponds to our partition $\lambda = (4,2,1,1,0)$, and the blue arrows display the described path. Taking the vertical step labels gives $k$-subset $r = \{ 1, 3, 4, 6, 9 \}$.

\begin{figure}[ht]
    \centering
    \begin{tikzpicture}[scale=.8]
\draw[<->] (-1,0) -- node[left] {$k=5$} (-1,4);
\draw[<->] (0,5) -- node[above] {$n-k=4$} (5,5);
\draw[dotted] (0,-1) -- (4,-1) -- (4,4) -- (0,4) -- (0,-1);
\draw (0,4) -- (4,4);
\draw (0,3) -- (4,3);
\draw (0,2) -- (1,2);
\draw (0,1) -- (1,1);
\draw (0,0) -- (1,0);
\draw (0,-1) -- (0,4);
\draw (1,0) -- (1,4);
\draw (2,2) -- (2,4);
\draw (3,3) -- (3,4);
\draw (4,3) -- (4,4);
 \draw[thick,blue,->] (0,-1) -- node[midway,right,font=\scriptsize, color=violet,xshift = -1mm]  {1} (0,0);
  \draw[thick,blue,->] (0,0)  -- node[midway,below,font=\scriptsize, color=violet,yshift = 0.5mm] {2} (1,0);
  \draw[thick,blue,->] (1,0)  -- node[midway,right,font=\scriptsize, color=violet,xshift = -1mm] {3} (1,1);
  \draw[thick,blue,->] (1,1)  -- node[midway,right,font=\scriptsize, color=violet,xshift = -1mm] {4} (1,2);
  \draw[thick,blue,->] (1,2)  -- node[midway,below,font=\scriptsize, color=violet,yshift = 1mm] {5} (2,2);
  \draw[thick,blue,->] (2,2)  -- node[midway,right,font=\scriptsize, color=violet,xshift = -1mm] {6} (2,3);
  \draw[thick,blue,->] (2,3)  -- node[midway,below,font=\scriptsize, color=violet,yshift = 1mm] {7} (3,3);
  \draw[thick,blue,->] (3,3)  -- node[midway,below,font=\scriptsize, color=violet,yshift = 1mm] {8} (4,3);
  \draw[thick,blue,->] (4,3)  -- node[midway,right,font=\scriptsize, color=violet,xshift = -1mm] {9} (4,4);
\end{tikzpicture}
    \caption{The Young diagram for partition $\lambda = (4,2,1,1,0)$. The labels of the blue path's vertical steps give the $k$-subset $r = \{ 1, 3, 4, 6, 9 \}$.}
    \label{fig: convert partition to subset}
\end{figure}
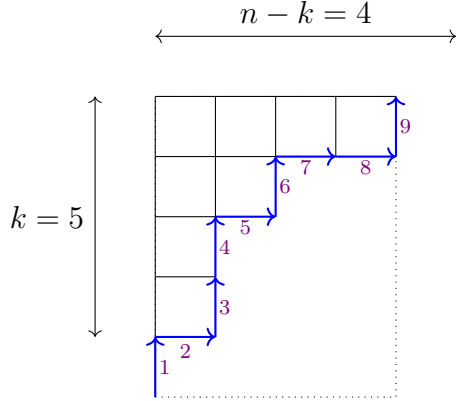

Algebraically, this correspondence is written
\begin{equation}\label{orderings}
r_i = \lambda_{(k+1-i)} +i.\end{equation}
The size of a partition $\lambda$ is defined as the number of boxes in its Young diagram, and is denoted $|\lambda |$. For a partition $\nu$ containing $\lambda$, the skew diagram $\nu/\lambda$ is called a \emph{rook strip} if each row and column contains at most one box. Equivalently, a rook strip may be defined as a skew diagram that is both a horizontal strip and a vertical strip. 
\vspace{.5cm}

The partitions of the $k \times (n-k)$ rectangle are ordered lexicographically and are labeled $\lambda^1, \dots, \lambda^{\binom{n}{k}}$. The corresponding $k$-subsets are labeled $r^1,\dots, r^{\binom{n}{k}}$.

\begin{example}
For $\textrm{Gr}(2,4)$ we have the following partitions and corresponding $k$-subsets:

\begin{align*}
    \lambda^1 = (0,0), \lambda^2 = (1,0), \lambda^3 &= (1,1), \lambda^4 =(2,0), \lambda^5 = (2,1), \lambda^6 = (2,2)\\
    r^1 = \{1,2\},r^2 = \{1,3\},r^3 &= \{2,3\},r^4 = \{1,4\},r^5 = \{2,4\},r^6 = \{3,4\}
\end{align*}
\end{example}
We also introduce the following notations: $\xb=(x_1,\dots,x_k),\ub=(u_1,\dots, u_n)$, and $\boldsymbol{d}=(d_1,\dots,d_k)$.

Over $\textrm{Gr}(k,n)$ there is a tautological sequence 
$$
0 \to \mathcal S \to \mathcal O_{\textrm{Gr}(k,n)}^{\oplus n} \to \mathcal Q \to 0
$$
where $\mathcal S$ is the tautological (sub)bundle and $\mathcal Q$ the tautological quotient bundle. We have that $\mathcal O (-1) = \det \mathcal S$ and $\mathcal O(1) = \det \mathcal Q$.

$\textrm{Gr}(k,n)$ admits a natural $T=(\C^\times)^n$-action induced by the action of $T$ on $\C^n$:
$$
(a_1, \ldots, a_n) \cdot (z_1, \ldots, z_n) = (a_1 z_1, \dots, a_n z_n).
$$
As a $T$-module, $\C^n$ decomposes as a direct sum of weight spaces $\C^n = \bigoplus_{i=1}^n \C_{u_i}$, where $u_i (a_1, \ldots, a_n) = a_i$. With this action we have equivariant versions of the above vector bundles. In this view, we can also view the $u_i$'s as equivariant Chern roots of the (equivariant) trivial bundle $\OC_{\mathrm{Gr}(k,n)}^{\oplus n}$. For the rest of the paper, $\OC(-1)$ will refer to the equivariant determinant bundle, unless otherwise stated.

\subsection{Cohomology}

Consider the complete flag $F_\bullet = (0 = F_0 \subset F_1 \subset \ldots \subset F_n = \C^n)$. The Schubert variety corresponding to a partition $\lambda = (\lambda_1 \geq \lambda_2 \geq \ldots \geq \lambda_k)$ is 
$$
X_\lambda (F_\bullet) = \{ V \in X \mid \dim (V \cap F_{k + i - \lambda_i}) \geq i \text{ for all $i \in [k]$} \},
$$
where $[k] := \{1,\dots, k\}$. The class of the Schubert variety is denoted by $\sigma_\lambda$, and does not depend on $F_\bullet$. The cohomology $\mathrm{H}^* (\mathrm{Gr}(k,n))$ (with $\Z$-coefficients) has a basis given by the Schubert classes $\sigma_\lambda$ for partitions $\lambda$ for partitions $\lambda$ in the $k \times (n-k)$ box. This basis will be referred to as the \emph{Schubert basis.}

The quantum cohomology of the Grassmannian $\textrm{Gr}(k,n)$ is an algebra, denoted as $\textrm{QH}^*(\textrm{Gr}(k,n))$. As a module, $\textrm{QH}^* (\textrm{Gr}(k,n)) = \textrm{H}^* (\textrm{Gr}(k,n)) \otimes_\Z \Z[z]$, where $z$ is the quantum parameter. The product structure of $\textrm{QH}^*(\textrm{Gr}(k,n))$ is given by a sum over partitions $\nu$ and degrees $d\geq 0$:
\begin{equation}
\label{eqn: quantum relations}
\sigma_\lambda * \sigma_\mu = \sum_{\nu,d \geq 0} z^d c_{\lambda,\mu} ^\nu \sigma_\nu
\end{equation}
where the coefficients $c^\nu_{\lambda, \mu}$ are given by certain genus 0 Gromov-Witten invariants (see \cite{BERTRAM1997289} and \cite{buch2001quantumcohomologygrassmannians}).

We will be mainly concerned with the equivariant versions of the above. The ring structure of $\textrm{QH}^*_T$ is defined very similarly to the usual quantum cohomology, except we replace $c^\nu_{\lambda, \mu}$ with certain genus 0 3-point equivariant Gromov-Witten invariants (see \cite{givental1996equivariantgromovwitten}). In this situation, $\mathrm{QH}_T^* (\mathrm{Gr}(k,n))$ will be an algebra over $\textrm{H}^*_T(\textrm{pt}) = \bb{Z}[u_1,\dots, u_n]$. $\textrm{H}^*_T (\textrm{Gr}(k,n))$ has a basis given by the equivariant versions of Schubert classes, which we also denote by $\sigma_\lambda$. Such classes can be represented as factorial Schur polynomials.

\subsubsection{Factorial Schur polynomials}
Let $x_1,\dots, x_k$ denote the Chern roots of the dual tautological subbundle, $\cal S^\vee$. It is known that $\textrm{H}^*_T(\textrm{Gr}(k,n))$ is isomorphic to the ring of polynomials symmetric in $x_1,\dots,x_k$ modulo certain relations. 

\begin{definition}
The \textit{factorial Schur polynomial} associated to the partition $\lambda$ in variables $\xb ,\ub$ is defined by the bialternant formula
\begin{equation} \label{factschur}
s_\lambda(\xb ;\ub) = \frac{\det\left( (x_i\, | \,\ub)^{\lambda_j+k-j} \right)_{1 \leq i, j \leq k}}{\prod_{i<j} (x_i-x_j)},
\end{equation}
where $ (x \,|\, \ub)^{b} =(x-u_1)\cdots (x-u_b)$ is the \emph{falling factorial}. 
\end{definition}

The equivariant Schubert classes can be represented as the factorial Schur polynomials. We may recover the ordinary Schur polynomials from these factorial Schur polynomials by sending each $u_i \to 0$.

The factorial Schur polynomials give the following presentation for the equivariant cohomology ring:
\begin{equation}
    \textrm{H}^*_T(\textrm{Gr}(k,n)) \cong \frac{\bb{Z}[\xb, \ub]^{S_k}}{\langle s_\lambda(\xb ;\ub) : \lambda \not \subseteq k \times (n-k) \text{ rectangle} \rangle}.
\end{equation}
where $\bb{Z}[\xb, \ub] ^{S_k}$ refers to the Laurent polynomials that are symmetric in the $x_1,\dots,x_k$, but not necessarily in the $u$'s.
For the equivariant quantum cohomology ring $\mathrm{QH}^*_T (\mathrm{Gr}(k,n))$, we add in the quantum multiplication relations given by equation \ref{eqn: quantum relations}. For more background on the equivariant quantum cohomology ring of the Grassmannian, we refer to \cite{buch2001quantumcohomologygrassmannians}, \cite{mihalcea2004equivariantquantumschubertcalculus}. 

\begin{example}
    Here are the factorial Schur polynomials corresponding to partitions of the $2\times 2$ box. For the Schur polynomials, take the limit $u_i \to 0$.
    \begin{align*}
        s_{(0,0)}(x_1,x_2) &= 1\\
        s_{(1,0)}(x_1,x_2) &= x_1-u_1+x_2-u_2\\
        s_{(1,1)}(x_1,x_2) &= (x_1-u_1)(x_2-u_1)\\
        s_{(2,0)}(x_1,x_2) &= \left( x_{{1}}-u_{{1}} \right)  \left( x_{{1}}-u_{{2}} \right) +
 \left( x_{{1}}-u_{{1}} \right)  \left( x_{{2}}-u_{{3}} \right) +
 \left( x_{{2}}-u_{{2}} \right)  \left( x_{{2}}-u_{{3}} \right) 
\\
        s_{(2,1)}(x_1,x_2) &= (x_1-u_1)(x_2-u_1)(x_1-u_2+x_2-u_3)\\
        s_{(2,2)}(x_1,x_2) &= (x_1-u_1)(x_1-u_2)(x_2-u_1)(x_2-u_2)\\
    \end{align*}
\end{example}

\subsection{K-theory}
To define the equivariant quantum K-theory ring for Grassmannians, we follow the treatment of \cite{mihalcealectures},\cite{Buch_2011}, \cite{gorbounov2025quantumktheorygrassmanniansyangbaxter}, with $z$ as the quantum parameter instead of $q$.
 
Denote the representation ring of $T$ by $\text{Rep}_T$. Abusing notation, we identify $u_i = [\bb{C}_{u_i}]$ and so we write the representation ring as the Laurent polynomial ring $\text{Rep}_T = \bb{Z}[u_1 ^{\pm 1} ,\dots, u_n ^{\pm 1}]$. These $u_i$'s are also the K-theoretic Chern roots of the trivial bundle.

The $T$-equivariant Grothendieck ring $\textrm{K}_T(\textrm{Gr}(k,n))$ is simply the extension by scalars of the usual K-theory ring, i.e. $\mathrm{K}_T (\mathrm{Gr}(k,n)) = \mathrm{K} ( \mathrm{Gr}(k,n)) \otimes_\Z \mathrm{Rep}_T$, with the multiplication extended linearly. $\mathrm{K}_T ( \mathrm{Gr}(k,n) )$ has a basis given by the $\binom{n}{k}$ equivariant Schubert structure sheaves $\OC_\lambda \coloneq [\OC_{X_\lambda}]$, for each $\lambda$ in the $k\times (n-k)$ rectangle. We refer to this basis as the \emph{Schubert basis} as well. Hence, as a $\mathrm{Rep}_T$-module, 
$$
\textrm{K}_T(\textrm{Gr}(k,n))=\bigoplus_\lambda \text{Rep}_T \OC_\lambda.
$$ 
The determinant tautological subbundle $\OC(-1)$ can be written as $\OC(-1) = 1-\OC_{(1)}.$
 
The $T$-equivariant quantum K-theory ring of the Grassmannian is an algebra over $\text{Rep}_T [[z]]$. As a $\text{Rep}_T [[z]]$ module,
$$
\mathrm{QK}_T(\mathrm{Gr}(k,n))  = \mathrm{K}_T(\mathrm{Gr}(k,n)) \otimes \mathrm{Rep}_T [[z]] = \bigoplus_{\lambda} \mathrm{Rep}_T [[z]] \OC_\lambda,
$$
where the sum is over all partitions $\lambda$ in the $k \times (n-k)$ rectangle. We will use $V \star W$ to denote the quantum multiplication of two elements $V, W \in \mathrm{QK}_T ( \mathrm{Gr}(k,n) )$. The quantum multiplication between basis elements is given as a sum 
$$
\OC_\lambda \star \OC_\mu = \sum \limits_{\nu, d \geq 0} N_{\lambda, \mu} ^{\nu, d} z^d \OC_\nu,
$$
where the structure constants $N^{v,d} _{\lambda, \mu}$ are given by K-theoretic equivariant genus 0 3-point Gromov-Witten invariants as described in \cite{Buch_2011} and \cite{givental1996equivariantgromovwitten}.

\subsubsection{Double Grothendieck polynomials}
Abusing notation, let $x_1,\dots,x_k$ be the K-theoretic Chern roots of the tautological subbundle. Similar to the case in cohomology, $\textrm{K}_T(\textrm{Gr}(k,n))$ is isomorphic to the ring of Laurent polynomials modulo certain relations \cite{cotti2024satakecorrespondenceequivariantquantum}. 
\begin{definition}
The \textit{double Grothendieck polynomials} associated to the partition $\lambda$ in variables $\xb,\ub$ is defined by the bialternant formula 
\begin{equation} \label{gdetformula}
    G_{\lambda} (\xb;\ub) = \frac{\det \left(x_j ^{i-1} (x_j \,|\, \ub) ^{\lambda_i +k -i} \right)_{1\leq i,j \leq k}}{\prod_{1\leq i<j\leq k }(x_j -x_i)},
\end{equation}
where $(x\, | \,\ub)$ is the \emph{K-theoretic falling factorial}, defined as $(x \,|\, \ub)^b = (1-x/u_1)\cdots(1-x/u_b).$
\end{definition}
While the notation is the same as the previously defined falling factorial, it will be clear depending on context. Also when clear, the Grothendieck polynomials may be denoted as $G_\lambda (\xb)$ or $G_\lambda$.

In $\mathrm{K}_T$, the representatives of the Schubert structure sheaves are given by the double Grothendieck polynomials. With these polynomials, we give the following presentation for the equivariant K-theory ring:
\begin{equation}
    \textrm{K}_T(\textrm{Gr}(k,n)) \cong \frac{\bb{Z}[\xb^\pm,\ub^\pm]^{S_k}}{\langle G_\lambda(\xb;\ub) : \lambda \not \subseteq k \times (n-k) \text{ rectangle} \rangle},
\end{equation}
where $\bb{Z}[\xb^\pm, \ub^\pm] ^{S_k}$ refers to the Laurent polynomials that are symmetric in the $x_1,\dots,x_k$, but not necessarily in the $u$'s.

\begin{example}
    Here are the double Grothendieck polynomials corresponding to partitions of the $2\times 2$ box. For the stable Grothendieck polynomials, send each $u_i \to 1$ for all $i$.
    \begin{align*}
        G_{(0,0)} (x_1,x_2)&= 1\\
        G_{(1,0)} (x_1,x_2)&= {1-\frac {x_{{1}}x_{{2}}}{u_{{1}}u_{{2}}}}\\
        G_{(1,1)} (x_1,x_2)&= {\frac { \left( u_{{1}}-x_{{1}} \right)  \left( u_{{1}}-x_{{2}}
 \right) }{{u_{{1}}}^{2}}}
\\
        G_{(2,0)}(x_1,x_2) &= {\frac {{x_{{1}}}^{2}x_{{2}}-x_{{2}} \left( u_{{1}}+u_{{2}}+u_{{3}}-x_
{{2}} \right) x_{{1}}+u_{{1}}u_{{2}}u_{{3}}}{u_{{1}}u_{{2}}u_{{3}}}}
\\
        G_{(2,1)} (x_1,x_2)&= {\frac { \left( u_{{1}}-x_{{1}} \right)  \left( u_{{1}}-x_{{2}}
 \right)  \left( u_{{2}}u_{{3}}-x_{{1}}x_{{2}} \right) }{{u_{{1}}}^{2}
u_{{2}}u_{{3}}}}
\\
        G_{(2,2)} (x_1,x_2)&= {\frac { \left( u_{{1}}-x_{{1}} \right)  \left( u_{{2}}-x_{{1}}
 \right)  \left( u_{{1}}-x_{{2}} \right)  \left( u_{{2}}-x_{{2}}
 \right) }{{u_{{1}}}^{2}{u_{{2}}}^{2}}}
    \end{align*}
\end{example}

In this paper we use a combinatorial identity for these double Grothendieck polynomials which holds for all inputs. This identity will be a key ingredient for the derivation of our formulas for the solutions of the K-theoretic QDE.
\begin{proposition}\label{pierigkprop}
Let $\lambda=(\lambda_1,\dots,\lambda_k)$ be a partition. We have the algebraic relation
    \begin{equation}
        x_1\cdots x_k G_\lambda  = u_{r_1}\cdots u_{r_k}\sum \limits_{\nu} (-1)^{|\lambda / \nu|} G_\nu ,
    \end{equation}
    where this sum is over all partitions $\nu$ (not necessarily in the $k\times (n-k)$ rectangle) such that $\nu \supset \lambda$ and $\nu/\lambda$ is a rook strip. $r$ is the corresponding $k$-subset to $\lambda$.
\end{proposition}
\begin{proof}
    This follows from section 4.2.1, line (25) of \cite{wheelerG}.
\end{proof}

\section{The quantum cohomology qde} \label{quantumcohomologysection}

The qde for the $\mathrm{QH}^*_T(\mathrm{Gr}(k,n))$ case is 
\begin{equation}
    \epsilon z\frac{\partial \psi(z)}{\partial z} = M^{C } (z) \psi(z),
\end{equation}
where $\epsilon \in \mathbb{C}^\times$. We will use a superscript $1n$ to denote solutions of this qde in the case of $\mathrm{Gr}(1,n) = \bb P^{n-1}$ and no superscript to indicate $\psi$ is a solution for $\mathrm{Gr}(k,n).$ 

We will use partition functions of the space to construct the solution. We will first define $M^C (z)$, and then obtain the $R$-matrix and partition functions. Then we will describe $\psi^{1n}$ and prove it is a solution to the qde. From here, we use the Satake correspondence as stated in Theorem 7.19 of \cite{cotti2024satakecorrespondenceequivariantquantum} to construct the solution $\psi$ to the qde for $\mathrm{QH}^*_T (\mathrm{Gr}(k,n))$.

\subsection{The operator \texorpdfstring{$M^C(z)$}{MC(z)}}
Let $M^C(z) := -c_1(\OC(-1)) *$, where $c_1(\OC(-1)) *$ is the operator of quantum multiplication by $c_1 (\mathcal O(-1))$. Note $M^C$ coincides with $\sigma_1 *$. The formula for this operator is a special case of the Pieri rule (\cite{mihalcea2004equivariantquantumschubertcalculus}) for $\sigma_1 *$: 
\begin{proposition} \label{cpieri}
Let $\mu =(\mu_1,\dots,\mu_k)$ be a partition in the $k \times (n-k)$ rectangle and $r=\{r_1,\dots,r_k\}$ be the corresponding $k$-subset. Then $$M^C(z) \sigma_\mu = \sum_\lambda \sigma_\lambda +\left(\sum _{i=1} ^k u_{r_i}\right)\sigma_\mu + z \sum_\nu \sigma_\nu ,$$ where $\lambda$ ranges over all partitions such that $|\lambda| = |\mu| +1$ and $$n-k \geq \lambda_1 \geq \mu_1 \geq \lambda_2 \geq \mu_2 \geq \cdots \geq \lambda_k \geq \mu_k,$$
and $\nu$ ranges over all partitions with $|\nu| = |\mu|+1-n$ and 
$$\mu_1-1 \geq \nu_1 \geq \mu_2-1 \geq \cdots \geq \mu_k -1 \geq \nu_k \geq 0.$$
\end{proposition}

With $M^C(z)$ defined, we obtain the $R$-matrix and the partition functions used to construct the solution to the qde.

\subsection{Partition functions for cohomology} \label{cohpar}
We find the $R$ matrix from a transition function of $\bb{P}^1$ as described in \cite{Okounkov2022-tf}. First consider the torus $\bb{C}_a ^\times$ acting on $\bb{P}^1$ as
$$[x:y] \to [ax:y] = [x: a^{-1}y].$$
We have the fixed points $p_1 = [1:0]=0$ and $p_2=[0:1]=\infty$.  We assign the ordering and weights as:
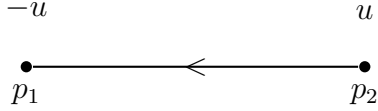
\[
\tikz[baseline=(p1.base), thick, >=stealth]{
  \node[circle, fill, inner sep=1.5pt, label=below:$p_1$] (p1) at (0,0) {};
  \node[circle, fill, inner sep=1.5pt, label=below:$p_2$] (p2) at (4.5,0) {};

  \draw[-] (p1) -- node[pos=0.5, yshift=-0.01ex] {${<}$} (p2);

  \node[above=4mm of p1] {$-u$};
  \node[above=4mm of p2] {$u$};
}
\]
Then we have attracting sets $Attr_{+} (p_1)=\bb{P}^1$ and $Attr_{+}(p_2)=p_2$. From this we have the positive attracting matrix defined with restrictions as $T_{ij}^+ = Attr_+ (p_j) |_{p_i}$:
\begin{equation}\label{firstp1}
    T^+ = \begin{pmatrix}
        Attr_+ (p_1) |_{p_1} &Attr_+ (p_2) |_{p_1}\\
        Attr_+ (p_1) |_{p_2}&Attr_+ (p_2) |_{p_2}
    \end{pmatrix}= \begin{pmatrix}
        1&0 \\ 1 &u
    \end{pmatrix}.
\end{equation}
Then consider the same diagram with the opposite order:
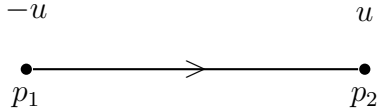
\[
\tikz[baseline=(p1.base), thick, >=stealth]{
  \node[circle, fill, inner sep=1.5pt, label=below:$p_1$] (p1) at (0,0) {};
  \node[circle, fill, inner sep=1.5pt, label=below:$p_2$] (p2) at (4.5,0) {};

  \draw[-] (p1) -- node[pos=0.5, yshift=-0.01ex] {${>}$} (p2);

  \node[above=4mm of p1] {$-u$};
  \node[above=4mm of p2] {$u$};
}
\]
Then we have $Attr_{-} (p_1)=p_1$ and $Attr_{-}(p_2)=\bb{P}^1$. From this we have
\begin{equation}\label{secp1}
    T^- = \begin{pmatrix}
        Attr_- (p_1) |_{p_1} &Attr_- (p_2) |_{p_1}\\
        Attr_- (p_1) |_{p_2}&Attr_- (p_2) |_{p_2}
        \end{pmatrix}= \begin{pmatrix}
        -u&1 \\ 0 &1
    \end{pmatrix}.
\end{equation}
Then we compute using (\ref{firstp1}) and (\ref{secp1}) to get
\begin{equation}
    (T^-)^{-1} \cdot T^+ = \begin{pmatrix}
        0&1\\
        1& u
    \end{pmatrix}.
\end{equation}

This forms the middle block of our $R$ matrix:
$$
R(u)=\begin{pmatrix}
    1&0&0&0\\
    0&0&1&0\\
    0&1&u&0\\
    0&0&0&1
\end{pmatrix}.$$

This satisfies the Yang-Baxter equation 
$$R_{a 1}(u_1-x) R_{a 2}(u_2-x) R_{12} (u_2-u_1)=R_{1 2}(u_2-u_1) R_{a 2}(u_2-x) R_{a 1}(u_1-x). $$
As written, $a$ denotes the auxiliary space and each Lax operator acts on the spaces of the subscript. For instance, $R_{a 1} : \bb{C}_a \otimes \bb{C}_{1}\otimes \bb{C}_{2}\to \bb{C}_a \otimes \bb{C}_{1}\otimes \bb{C}_{2} $ acts as $R \otimes I$, only on the first two and as the identity on the third space.
With these weights, the partition function can be obtained via lattice models. This was done in \cite{aggaC}, and is the same as the factorial Schur polynomial.

\begin{proposition}
The partition function for $\mathrm{H}^*_T(\textrm{Gr}(k,n))$ coincides with the factorial Schur polynomials:
    $$s_\lambda(\xb;\ub) = \frac{\det\left( (x_i\, | \,\ub)^{\lambda_j+k-j} \right)_{1 \leq i, j \leq k}}{\prod_{i<j} (x_i-x_j)},$$
where $ (x \,|\, \ub)^{b} =(x-u_1)\cdots (x-u_b)$ is the falling factorial.
\end{proposition}

\subsection{\texorpdfstring{$\textrm{QH}^*_T(\textrm{Gr}(1,n))$}{QHT} case: solution to qde}
We prove the solution to the qde 
$\epsilon z {\partial }_z \psi(z)^{1n} = M(z)^{C} \psi(z)^{1n}$, using products of Gamma functions to construct the explicit form of the solution. Recall the identity $\Gamma(x+1) = x\Gamma(x)$.

\begin{proposition} \label{qdesolnceq1n}
    The solution the qde for $\mathrm{QH}^*_T(\mathrm{Gr}(1,n))$ is as follows: choose a partition $\lambda$ in the $1\times (n-1)$ rectangle. Then the solution $\psi_\lambda^{1n}$ is a class in $\mathrm{QH}^*_T(\mathrm{Gr}(1,n))$ of the form:

    \begin{equation}
        \psi_\lambda ({x}, \ub, z)^{1n} = z^{\frac{x}{\epsilon}} \sum_{d=0} ^\infty \frac{1}{\prod \limits_{\ell=1} ^n \Gamma \left(\frac{x}{\epsilon}+d-\frac{u_\ell}{\epsilon}+1\right) }s_\lambda (x+d \epsilon) \left(\frac{z}{\epsilon^n}\right)^d
    \end{equation}
\end{proposition}
\begin{proof}
Out of the possible partitions, we have two cases. For convenience we will omit writing $x$ and $\ub$ in $\psi(x,\ub,z)$ for the proof.
\vspace{.5cm}

\emph{Case 1:}  $\lambda \neq (n-1)$
Recall that with Proposition \ref{cpieri}, $M(z)^C s_\lambda = u_{r_1} s_\lambda + s_{\lambda+1}$, where $r$ is the corresponding $1$-subset to $\lambda$ (in this case $\lambda_1+1=r_1$). Then 
$$M(z)^C\psi_\lambda (z)^{1n} =  z^{\frac{x}{\epsilon}} \sum_{d=0} ^\infty \frac{ u_{r_1} s_\lambda (x+d \epsilon) +s_{\lambda+1} (x+d \epsilon) }{\prod_{\ell=1} ^n \Gamma \left(\frac{x}{\epsilon}+d-\frac{u_\ell}{\epsilon}+1\right)}\left(\frac{z}{\epsilon^n}\right)^d.$$
With this, the difference of both sides is
\begin{align}
&\epsilon z \frac{\partial }{\partial z} \psi_\lambda(z)^{1n} - M(z)^{C} \psi_\lambda(z)^{1n} \nonumber\\
&= \epsilon z \cdot   \sum_{d=0} ^\infty \frac{(x/\epsilon +d)\, s_\lambda (x+d \epsilon) }{\prod_{\ell=1} ^n \Gamma \left(\frac{x}{\epsilon}+d-\frac{u_\ell}{\epsilon}+1\right)} z^{\frac{x}{\epsilon}+d-1}{\epsilon^{-nd}}     \\
&-z^{\frac{x}{\epsilon}}     \sum_{d=0} ^\infty \frac{ u_{r_1} s_\lambda (x+d \epsilon) +s_{\lambda+1} (x+d \epsilon) }{\prod_{\ell=1} ^n \Gamma \left(\frac{x}{\epsilon}+d-\frac{u_\ell}{\epsilon}+1\right)}\left(\frac{z}{\epsilon^n}\right)^d \nonumber \\
&=z^{\frac{x}{\epsilon}}     \sum_{d=0} ^\infty \frac{ (x+d\epsilon -u_{r_1}) s_\lambda (x+d \epsilon) -s_{\lambda+1} (x+d \epsilon) }{\prod_{\ell=1} ^n \Gamma \left(\frac{x}{\epsilon}+d-\frac{u_\ell}{\epsilon}+1\right)}\left(\frac{z}{\epsilon^n}\right)^d.\label{clsm}
\end{align}
By the factorial Schur polynomial formula (\ref{factschur}), we have $s_\lambda(x) = \prod_{m=1}^{\lambda_1}(x-u_m)$. From this, the numerator of (\ref{clsm}) simplifies as

\begin{align*}
    (x\!+\!d\epsilon\!-\!u_{r_1}) s_\lambda (x\!+\!d \epsilon) \ -\  &s_{\lambda+1} (x\!+\!d \epsilon)= (x\!+\!d\epsilon\!-\!u_{r_1}) \prod_{m=1}^{\lambda_1}(x\!+\!d\epsilon\!-\!u_m) -\prod_{m=1}^{\lambda_1+1}(x\!+\!d\epsilon\!-\!u_m)\\
    &=(x\!+\!d\epsilon\!-\!u_{r_1}) \prod_{m=1}^{\lambda_1}(x\!+\!d\epsilon\!-\!u_m) -(x\!+\!d\epsilon\!-\!u_{\lambda_1+1}) \prod_{m=1}^{\lambda_1}(x\!+\!d\epsilon\!-\!u_m)\\
    &=(x\!+\!d\epsilon\!-\!u_{r_1}) \prod_{m=1}^{\lambda_1}(x\!+\!d\epsilon\!-\!u_m) -(x\!+\!d\epsilon\!-\!u_{r_1}) \prod_{m=1}^{\lambda_1}(x\!+\!d\epsilon\!-\!u_m)\\
    &=0,
\end{align*}
which results in (\ref{clsm}) reducing to 0.
\vspace{0.5cm}
  
\emph{Case 2:} $\lambda = (n-1)$

The beginning is similar to case 1, but note $M(z)^C s_{n-1}=u_n s_\lambda +z$, and we will show the difference of both sides telescopes to 0. Observe
\begin{align}
      &\epsilon z \frac{\partial }{\partial z} \psi_{\lambda}(z)^{1n} - M(z)^{C} \psi_\lambda(z)^{1n} \nonumber\\
      &= \epsilon z \cdot   \sum_{d=0} ^\infty \frac{(x/\epsilon +d)\, s_{n-1} (x+d \epsilon) }{\prod_{\ell=1} ^n \Gamma \left(\frac{x}{\epsilon}+d-\frac{u_\ell}{\epsilon}+1\right)} z^{\frac{x}{\epsilon}+d-1}{\epsilon^{-nd}}-        z^{\frac{x}{\epsilon}}     \sum_{d=0} ^\infty \frac{ u_{n} s_{n-1} (x+d \epsilon) +z }{\prod_{\ell=1} ^n \Gamma \left(\frac{x}{\epsilon}+d-\frac{u_\ell}{\epsilon}+1\right)}\left(\frac{z}{\epsilon^n}\right)^d \nonumber \\
      &=z^{\frac{x}{\epsilon}}     \sum_{d=0} ^\infty \frac{ (x+d\epsilon-u_n) s_{n-1} (x+d \epsilon) -z }{\prod_{\ell=1} ^n \Gamma \left(\frac{x}{\epsilon}+d-\frac{u_\ell}{\epsilon}+1\right)}\left(\frac{z}{\epsilon^n}\right)^d.\label{case2sum}
  \end{align}
We have from the factorial Schur polynomial formula (\ref{factschur}) that
\begin{align*}
    (x+d\epsilon-u_n) s_{n-1} (x+d \epsilon)&= (x+d\epsilon-u_n) \prod_{m=1} ^{n-1} (x+d\epsilon - u_m)\\
    &= \prod_{m=1} ^{n} (x+d\epsilon - u_m).
\end{align*}
Note that when $d=0$, restricting $x$ to any fixed point (i.e. $x=u_i$ for any $i\in [n]$) makes this product 0. By equivariant localization of the ring, this quantity is equivalent to 0 in this ring. Using this, we further split (\ref{case2sum}) into two terms and eliminate the $d=0$ term of the first sum:
\begin{align}
    z^{\frac{x}{\epsilon}}     &\sum_{d=0} ^\infty \frac{ (x+d\epsilon-u_n) s_{n-1} (x+d \epsilon) -z }{\prod_{\ell=1} ^n \Gamma \left(\frac{x}{\epsilon}+d-\frac{u_\ell}{\epsilon}+1\right)}\left(\frac{z}{\epsilon^n}\right)^d = \nonumber\\
    &=z^{\frac{x}{\epsilon}}     \sum_{d=0} ^\infty \frac{ \prod_{m=1} ^{n} (x+d\epsilon - u_m) }{\prod_{\ell=1} ^n \Gamma \left(\frac{x}{\epsilon}+d-\frac{u_\ell}{\epsilon}+1\right)}\left(\frac{z}{\epsilon^n}\right)^d - z^{\frac{x}{\epsilon}}     \sum_{d=0} ^\infty \frac{ z }{\prod_{\ell=1} ^n \Gamma \left(\frac{x}{\epsilon}+d-\frac{u_\ell}{\epsilon}+1\right)}\left(\frac{z}{\epsilon^n}\right)^d \nonumber\\
    &=z^{\frac{x}{\epsilon}}     \sum_{d=1} ^\infty \frac{ \prod_{m=1} ^{n} (x+d\epsilon - u_m) }{\prod_{\ell=1} ^n \Gamma \left(\frac{x}{\epsilon}+d-\frac{u_\ell}{\epsilon}+1\right)}\left(\frac{z}{\epsilon^n}\right)^d - z^{\frac{x}{\epsilon}}     \sum_{d=0} ^\infty \frac{ z }{\prod_{\ell=1} ^n \Gamma \left(\frac{x}{\epsilon}+d-\frac{u_\ell}{\epsilon}+1\right)}\left(\frac{z}{\epsilon^n}\right)^d. \label{case22}
\end{align}
Using the shift identity of the Gamma function, we have $$\frac{(x+d\epsilon -u_i)}{\Gamma(\frac{x}{\epsilon} +d -\frac{u_i}{\epsilon}+1)} = \frac{\epsilon}{ \Gamma(\frac{x}{\epsilon} +d -\frac{u_i}{\epsilon})}.$$ Then the first term of (\ref{case22}) changes:
\begin{align*}
    &=z^{\frac{x}{\epsilon}}     \sum_{d=1} ^\infty \frac{ \epsilon^n }{\prod_{\ell=1} ^n \Gamma \left(\frac{x}{\epsilon}+d-\frac{u_\ell}{\epsilon}\right)}\left(\frac{z}{\epsilon^n}\right)^d - z^{\frac{x}{\epsilon}}     \sum_{d=0} ^\infty \frac{ z }{\prod_{\ell=1} ^n \Gamma \left(\frac{x}{\epsilon}+d-\frac{u_\ell}{\epsilon}+1\right)}\left(\frac{z}{\epsilon^n}\right)^d\\
    &=z^{\frac{x}{\epsilon}}     \sum_{d=1} ^\infty \frac{ 1 }{\prod_{\ell=1} ^n \Gamma \left(\frac{x}{\epsilon}+d-\frac{u_\ell}{\epsilon}\right)}\left(\frac{z^d}{\epsilon^{n(d-1)}}\right) - z^{\frac{x}{\epsilon}}     \sum_{d=0} ^\infty \frac{ 1 }{\prod_{\ell=1} ^n \Gamma \left(\frac{x}{\epsilon}+d-\frac{u_\ell}{\epsilon}+1\right)}\left(\frac{z^{d+1}}{\epsilon^{nd}}\right).
\end{align*}
Re-indexing the first sum with shift $d\to d+1$ results in the difference of two identical sums, completing the proof.
\end{proof}

\subsection{\texorpdfstring{$\textrm{QH}^*_T(\textrm{Gr}(k,n))$}{QHT} case: solution to qde via geometric Satake}

By \cite{cotti2024satakecorrespondenceequivariantquantum}, there is an isomorphism of $\Z[z]$-modules
$$
\vartheta_{k,n} : \bigwedge^k_{\Z[z]} \mathrm{H}^*_T (\bb P^{n-1}) \to \mathrm{H}_T^* ( \mathrm{Gr}(k,n) ).
$$
Let $\sigma^{\bb P}_p$ denote a class in the cohomology of $\bb P^{n-1}$ and $\sigma^G_\lambda$ denote a Schubert class in the cohomology of the Grassmannian. For $\lambda = (\lambda_1, \ldots, \lambda_k)$, the isomorphism $\vartheta_{k,n}$ takes
$$
\sigma^{\bb P}_{\lambda_k} \wedge \sigma^{\bb P}_{\lambda_{k-1} + 1} \wedge \ldots \wedge \sigma^{\bb P}_{\lambda_1 +k - 1} \mapsto \sigma^G_\lambda.
$$
By solving the qde for $\bb P^{n-1}$, we get solutions that live in $\mathrm{QH}_T^* (\bb P^{n-1})$. If we take the $k$-th wedge power and then map by $\vartheta_{k,n}$, \cite{cotti2024satakecorrespondenceequivariantquantum} showed that the result is now a solution to the qde for $\mathrm{Gr}(k,n)$. Using this Satake correspondence, we are able to construct the solution to the qde in $\mathrm{QH}^*_T(\mathrm{Gr}(k,n))$ from copies of the qde solution for $\mathrm{QH}^*_T(\mathrm{Gr}(1,n))$.

Denote by $\db = (d_1,\dots,d_k)$ and $\xb+\db \epsilon := (x_1+d_1\epsilon, \dots, x_k+d_k \epsilon)$.

\begin{theorem} \label{qhtknthm}
    The solution the qde for $\mathrm{QH}^*_T(\mathrm{Gr}(k,n))$ is as follows: choose a partition $\lambda$ in the $k\times (n-k)$ rectangle. Then the solution $\psi_\lambda$ is a class in $\mathrm{QH}^*_T(\mathrm{Gr}(k,n))$ taking the form:

    \begin{equation}\label{psiCeq}
        \psi_\lambda ({\xb}, \ub, z) = C \sum_{\db} \frac{\prod \limits_{b,c=1}^k \left(\frac{x_c-x_b}{\epsilon}+d_c-d_b+1\right) }{\prod \limits_{b=1} ^k \prod \limits_{\ell=1} ^n \Gamma \left(\frac{x_b}{\epsilon}+d_b-\frac{u_\ell}{\epsilon}+1\right)} s_\lambda (\xb+\db \epsilon) \left(\frac{z}{\epsilon^n}\right)^{|\db|}
    \end{equation}
    where $C= ((-1)^{k-1} z)^{(x_1+\cdots x_k)/\epsilon} \prod_{b,c=1} ^k \left(\frac{x_c}{\epsilon} - \frac{x_b}{\epsilon}  +1\right)^{-1}$ The sum over $\db=(d_1,\dots,d_k)$ means to sum over $d_b \geq 0$ for each $b\in [k]$, and $|\db|=d_1+\cdots d_k$.
\end{theorem}

In order to prove this Theorem we first need the following Lemma:
\begin{lemma} \label{lemmaratioceq}
    \begin{equation}
        \prod \limits_{b=1}^k \prod \limits_{c=1} ^k \frac{\Gamma((\frac{x_c}{\epsilon}+d_c)-(\frac{x_b}{\epsilon}+d_b)+1)}{\Gamma(\frac{x_c}{\epsilon}-\frac{x_b}{\epsilon}+1)}=(-1)^{(k+1)|\db|}\prod \limits_{1\leq i<j\leq k} \frac{ (\frac{x_i}{\epsilon}+d_i)-(\frac{x_j}{\epsilon}+d_j)}{\frac{x_i}{\epsilon}-\frac{x_j}{\epsilon}} 
    \end{equation}
\end{lemma}
\begin{proof}
First note that in the left side's double product, the $b=c$ terms are 1, so we may write 
\begin{align}
     \prod \limits_{b=1}^k \prod \limits_{c=1} ^k &\frac{\Gamma((\frac{x_c}{\epsilon}+d_c)-(\frac{x_b}{\epsilon}+d_b)+1)}{\Gamma(\frac{x_c}{\epsilon}-\frac{x_b}{\epsilon}+1)} = \nonumber \\
     &=\left(\prod \limits_{b<c} \frac{\Gamma((\frac{x_c}{\epsilon}+d_c)-(\frac{x_b}{\epsilon}+d_b)+1)}{\Gamma(\frac{x_c}{\epsilon}-\frac{x_b}{\epsilon}+1)}\right)\left(\prod \limits_{b>c} \frac{\Gamma((\frac{x_c}{\epsilon}+d_c)-(\frac{x_b}{\epsilon}+d_b)+1)}{\Gamma(\frac{x_c}{\epsilon}-\frac{x_b}{\epsilon}+1)}\right) \nonumber \\
     &=\prod \limits_{b<c} \left(\frac{\Gamma((\frac{x_c}{\epsilon}+d_c)-(\frac{x_b}{\epsilon}+d_b)+1)}{\Gamma(\frac{x_c}{\epsilon}-\frac{x_b}{\epsilon}+1)}\cdot \frac{\Gamma((\frac{x_b}{\epsilon}+d_b)-(\frac{x_c}{\epsilon}+d_c)+1)}{\Gamma(\frac{x_b}{\epsilon}-\frac{x_c}{\epsilon}+1)}\right). \label{cohlemmalastline}
\end{align}
Now pick $\xi,\eta \in [k]$. Note that using the identity $\Gamma(x+1)=x\Gamma(x)$ we find each term of the product can be written as 
\begin{equation*}
\frac{\Gamma((\frac{x_\xi}{\epsilon}+d_\xi)-(\frac{x_\eta}{\epsilon}+d_\eta)+1)}{\Gamma(\frac{x_\xi}{\epsilon}-\frac{x_\eta}{\epsilon}+1)}\frac{\Gamma((\frac{x_\eta}{\epsilon}+d_\eta)-(\frac{x_\xi}{\epsilon}+d_\xi)+1)}{\Gamma(\frac{x_\eta}{\epsilon}-\frac{x_\xi}{\epsilon}+1)} =\left(\frac{\frac{x_\xi-x_\eta}{\epsilon}+d_\xi-d_\eta}{\frac{x_\xi-x_\eta}{\epsilon}}\right) (-1)^{d_\xi-d_\eta}.
\end{equation*}
 Note that $\prod_{i<j} (-1)^{d_i-d_j}= (-1)^{(k+1)|\db|}$, so we use the this and last line to write (\ref{cohlemmalastline}) as
$$
(-1)^{(k+1)|\db|}\prod \limits_{1\leq i<j\leq k} \frac{ (\frac{x_i}{\epsilon}+d_i)-(\frac{x_j}{\epsilon}+d_j)}{\frac{x_i}{\epsilon}-\frac{x_j}{\epsilon}} ,
$$
completing the proof.

\end{proof}

\begin{proof}[Proof of Theorem \ref{qhtknthm}]
Using the Satake Correspondence from \cite{cotti2024satakecorrespondenceequivariantquantum}, the solution $\psi  (z)$ is a wedge product of solutions $\psi^{1n}$, with a shift of $z$ and dividing by the Vandermonde determinant. Written as a determinant we have 

\begin{equation}\label{ceqinitialdet}
    \psi_\lambda(\xb, \ub, z) = \frac{\det \left[ \psi^{1n}_{\lambda_i+k-i} \left(x_j, \ub, (-1)^{k-1} z\right) \right]_{1\leq i,j \leq k}}{\prod_{1\leq i<j \leq k} (x_i-x_j)}.
\end{equation}
We write this in our desired form via algebraic manipulation. First we rewrite the numerator. Using the definition of determinant, the numerator of (\ref{ceqinitialdet}) is
\begin{equation}\label{ceqline2}
    =\sum_{\sigma \in S_k} \mathrm{sgn}(\sigma) \prod_{j=1} ^k \psi^{1n}_{\lambda_{\sigma(j)}+k-{\sigma(j)}} \left(x_j, (-1)^{k-1} z\right).
\end{equation}
By Proposition \ref{qdesolnceq1n}, we write (\ref{ceqline2}) as 
\begin{equation}\label{ceqline3}
    =\sum_{\sigma \in S_k} \mathrm{sgn}(\sigma) \prod_{j=1} ^k \left(((-1)^{k-1}z)^{\frac{x_j}{\epsilon}}\sum_{d_j=0} \frac{s_{\lambda_{\sigma(j)}+k-{\sigma(j)}} (x_j+d_j \epsilon)}{ \prod_{\ell=1} ^n \Gamma \left(\frac{x_j}{\epsilon}+d_j- \frac{u_\ell}{\epsilon}+1\right)}  \left(\frac{(-1)^{k-1} z}{\epsilon^n}\right)^{d_j}  \right).
\end{equation}
Let $\widetilde{C} = \prod_{j=1} ^k ((-1)^{k-1} z)^{x_j/\epsilon}= ((-1)^{k-1}z)^{(x_1+\cdots + x_k)/\epsilon}$. Since it doesn't depend on either index, we will factor it out. For our sum we use the more general form of the Cauchy product to write the product of sums (\ref{ceqline3}) as a sum of products over $\db$:
\begin{equation}
    \widetilde{C}\sum_{\sigma \in S_k} \mathrm{sgn}(\sigma) \sum_{\db} \prod_{j=1} ^k \left(\frac{s_{\lambda_{\sigma(j)}+k-{\sigma(j)}} (x_j+d_j \epsilon)}{ \prod_{\ell=1} ^n \Gamma \left(\frac{x_j}{\epsilon}+d_j- \frac{u_\ell}{\epsilon}+1\right)}  \left(\frac{(-1)^{k-1} z}{\epsilon^n}\right)^{d_j}  \right).
\end{equation}

We'll switch the sums and distribute the product to write this as 
\begin{equation}
     \widetilde{C} \sum_{\db}  \sum_{\sigma \in S_k} \mathrm{sgn}(\sigma)\prod_{j=1} ^k \left(\frac{s_{\lambda_{\sigma(j)}+k-{\sigma(j)}} (x_j+d_j \epsilon)}{ \prod_{\ell=1} ^n \Gamma \left(\frac{x_j}{\epsilon}+d_j- \frac{u_\ell}{\epsilon}+1\right)}  \right) \left (\prod_{j=1} ^k\left(\frac{(-1)^{k-1} z}{\epsilon^n}\right)^{d_j}  \right).
\end{equation}
From $\prod_{j=1}^k \left({(-1)^{k-1} z}/{\epsilon^n}\right)^{d_j} = \left({(-1)^{k-1} z}/{\epsilon^n}\right)^{|\db|}$, we have
\begin{equation}\label{ceqline4}
    \widetilde{C} \sum_{\db}  \sum_{\sigma \in S_k} \mathrm{sgn}(\sigma)\prod_{j=1} ^k \left(\frac{s_{\lambda_{\sigma(j)}+k-{\sigma(j)}} (x_j+d_j \epsilon)}{ \prod_{\ell=1} ^n \Gamma \left(\frac{x_j}{\epsilon}+d_j- \frac{u_\ell}{\epsilon}+1\right)}  \right) \left(\frac{(-1)^{k-1} z}{\epsilon^n}\right)^{|\db|}.
\end{equation}
Since ${\lambda_{\sigma(j)}+k-{\sigma(j)}}$ is an integer, the factorial Schur polynomial formula gives $$
s_{\lambda_{\sigma(j)}+k-{\sigma(j)}}(x) = \prod_{m=1} ^ {\lambda_{\sigma(j)}+k-{\sigma(j)}} (x-u_m).
$$
With this, (\ref{ceqline4}) becomes
\begin{equation}\label{ceqline5}
    \widetilde{C}  \sum_{\db} \frac{\sum_{\sigma \in S_k} \mathrm{sgn}(\sigma)\prod_{j=1} ^k \prod_{m=1} ^ {\lambda_{\sigma(j)}+k-{\sigma(j)}} (x_j+d_j \epsilon-u_m)}{\prod_{j=1} ^k  \prod_{\ell=1} ^n \Gamma \left(\frac{x_j}{\epsilon}+d_j- \frac{u_\ell}{\epsilon}+1\right)} \left(\frac{(-1)^{k-1} z}{\epsilon^n}\right)^{|\db|}. 
\end{equation}
Observe by manipulating the numerator, we obtain
\begin{align*}
    \sum_{\sigma \in S_k} \mathrm{sgn}(\sigma)\prod_{j=1} ^k \prod_{m=1} ^ {\lambda_{\sigma(j)}+k-{\sigma(j)}} (x_j+d_j\epsilon-u_m) &= \sum_{\sigma \in S_k} \mathrm{sgn}(\sigma)\prod_{j=1} ^k  (x_j +d_j\epsilon \, | \, \ub)^{\lambda_{\sigma(j)}+k-{\sigma(j)}}\\
    &= \det \left[(x_j +d_j\epsilon \, | \, \ub)^{\lambda_{i}+k-{i}}\right]_{1\leq i,j \leq k}\\
    &= s_\lambda (\xb+\db \epsilon) \, \prod_{1\leq i <j \leq k} (x_i+d_i\epsilon - x_j - d_j \epsilon).
\end{align*}
With this manipulation, we write (\ref{ceqline5}) as 
\begin{equation}\label{lastofnum}
    \widetilde{C}  \sum_{\db} \frac{s_\lambda (\xb+\db \epsilon) \, \prod_{1\leq i <j \leq k} (x_i+d_i\epsilon - x_j - d_j \epsilon)}{\prod_{j=1} ^k  \prod_{\ell=1} ^n \Gamma \left(\frac{x_j}{\epsilon}+d_j- \frac{u_\ell}{\epsilon}+1\right)} \left(\frac{(-1)^{k-1} z}{\epsilon^n}\right)^{|\db|}. 
\end{equation}
Now that the numerator has been rewritten, we will divide back the denominator. Using (\ref{lastofnum}), we write (\ref{ceqinitialdet}) as
\begin{equation}\label{numden1}
     \widetilde{C}  \sum_{\db} \frac{s_\lambda (\xb+\db \epsilon) \, }{\prod_{j=1} ^k  \prod_{\ell=1} ^n \Gamma \left(\frac{x_j}{\epsilon}+d_j- \frac{u_\ell}{\epsilon}+1\right)}\prod_{1\leq i <j \leq k} \frac{(x_i+d_i\epsilon - x_j - d_j \epsilon)}{x_i-x_j} \left(\frac{(-1)^{k-1} z}{\epsilon^n}\right)^{|\db|}. 
\end{equation}
Note that Lemma \ref{lemmaratioceq} gives us
\begin{align*}
    \prod_{1\leq i <j \leq k} \frac{(x_i+d_i\epsilon - x_j - d_j \epsilon)}{x_i-x_j} &= \prod_{1\leq i <j \leq k} \frac{(\frac{x_i}{\epsilon}+d_i - \frac{x_j}{\epsilon} - d_j )}{\frac{x_i}{\epsilon}-\frac{x_j}{\epsilon}}\\
    &=(-1)^{(k-1)|\db|}\prod_{b,c=1} ^k \frac{\Gamma \left(\frac{x_c}{\epsilon}+d_c - \frac{x_b}{\epsilon} - d_b+1 \right)}{\Gamma \left(\frac{x_c}{\epsilon} - \frac{x_b}{\epsilon}  +1\right)}.
\end{align*}
Plugging this manipulation back into (\ref{numden1}) and simplifying we may cancel the $(-1)^{(k-1)|\db|}$ terms. In addition let $C=\widetilde{C} \prod_{b,c=1} ^k \left(\frac{x_c}{\epsilon} - \frac{x_b}{\epsilon}  +1\right)^{-1}$, Then (\ref{numden1}) is
\begin{equation}
    C  \sum_{\db} \frac{\prod_{b,c=1} ^k {\Gamma \left(\frac{x_c}{\epsilon}+d_c - \frac{x_b}{\epsilon} - d_b+1 \right)}  }{\prod_{j=1} ^k  \prod_{\ell=1} ^n \Gamma \left(\frac{x_j}{\epsilon}+d_j- \frac{u_\ell}{\epsilon}+1\right)}s_\lambda (\xb+\db \epsilon)\left(\frac{ z}{\epsilon^n}\right)^{|\db|}, 
\end{equation}
completing the proof.
\end{proof}

\subsection{Bethe equations and Bethe vectors}
Using the solution to the qde for $\mathrm{QH}^*_T(\mathrm{Gr}(k,n))$, we asymptotically can derive the Bethe equations. These coincide with those derived from the Bethe Ansatz (see Section 4.2 of \cite{gorbkorff2017}).

\begin{proposition}
    For $\mathrm{QH}^*_T(\mathrm{Gr}(k,n))$ there are $k$ Bethe equations: One for each $j\in [k]$ given by
    $$
    (x_j-u_1)\cdots (x_j-u_n) = (-1)^{k-1}z.
    $$
\end{proposition}

\begin{proof}
First note that as $\epsilon \to 0$ we have $\partial_x \ln \Gamma(x/\epsilon)= \Gamma'(x/\epsilon)/\Gamma(x/\epsilon) \approx \ln(x/\epsilon)=\ln(x)-\ln(\epsilon)$.

Let $\varphi$ denote the integrand solution except the factorial Schur polynomial and part of $C$:
$$\varphi :=\frac{\prod_{b,c=1} ^k {\Gamma \left(\frac{x_c}{\epsilon}+d_c - \frac{x_b}{\epsilon} - d_b+1 \right)}  }{\prod_{j=1} ^k  \prod_{\ell=1} ^n \Gamma \left(\frac{x_j}{\epsilon}+d_j- \frac{u_\ell}{\epsilon}+1\right)}\left(\frac{ z^{|\db|+(x_1+\cdots x_k)/\epsilon} }{\epsilon^{|\db| n}}\right). $$ The saddle point equations are determined by $x_j\partial_{x_j} \ln (\varphi)=0$ for $j\in [k]$.

For $j\in [k]$ the left side is 
\begin{align*}
    &x_j \partial_{x_j} \ln \left(\frac{\prod_{b\neq c} {\Gamma \left(\frac{x_c}{\epsilon}+d_c - \frac{x_b}{\epsilon} - d_b+1 \right)}  }{\prod_{b=1} ^k  \prod_{\ell=1} ^n \Gamma \left(\frac{x_b}{\epsilon}+d_b- \frac{u_\ell}{\epsilon}+1\right)}\left(\frac{ z^{|\db|+(x_1+\cdots x_k)/\epsilon} }{\epsilon^{|\db| n}}\right)\right)\\
    &= x_j \partial_{x_j}\sum_{b\neq c}   \ln \Gamma \left(\frac{x_c}{\epsilon}+d_c - \frac{x_b}{\epsilon} - d_b+1 \right) - x_j \partial_{x_j}\sum_{b=1} ^k \sum_{\ell=1}^n  \ln \Gamma \left(\frac{x_b}{\epsilon}+d_b- \frac{u_\ell}{\epsilon}+1\right)+\\
    &\hspace{4in}+ x_j \partial_{x_j} \ln \left(\frac{ z^{|\db|+(x_1+\cdots x_k)/\epsilon} }{\epsilon^{|\db| n}}\right)\\
    &=\sum_{b\neq j} ^k x_j \partial_{x_j}    \ln \Gamma \left(\frac{x_j}{\epsilon}+d_j - \frac{x_b}{\epsilon} - d_b+1 \right) + \sum_{c\neq j} ^k x_j \partial_{x_j}    \ln \Gamma \left(\frac{x_c}{\epsilon}+d_c - \frac{x_j}{\epsilon} - d_j+1 \right) \\
    &\hspace{1in}-  \sum_{\ell=1}^n   x_j \partial_{x_j} \ln \Gamma \left(\frac{x_j}{\epsilon}+d_j- \frac{u_\ell}{\epsilon}+1\right)+  x_j \partial_{x_j} \ln \left(\frac{ z^{|\db|+(x_1+\cdots x_k)/\epsilon} }{\epsilon^{|\db| n}}\right). 
\end{align*}
Note as $\epsilon \to 0$, we have $\partial_{x_j} \ln \Gamma \left(\frac{x_j}{\epsilon}+d_j - \frac{x_b}{\epsilon} - d_b+1 \right)\approx\partial_{x_j} \ln \Gamma \left(\frac{x_j}{\epsilon} - \frac{x_b}{\epsilon}  \right)\approx \frac{1}{\epsilon}\ln(\frac{x_j-x_b}{\epsilon})$. Similarly, $\partial_{x_j} \ln \Gamma \left(\frac{x_c}{\epsilon}+d_c - \frac{x_j}{\epsilon} - d_j+1 \right)\approx -\frac{1}{\epsilon}\ln(\frac{x_c-x_j}{\epsilon})$. For the third sum, we have that $\partial_{x_j} \ln \Gamma \left(\frac{x_j}{\epsilon}+d_j- \frac{u_\ell}{\epsilon}+1\right) \approx \partial_{x_j} \ln \Gamma \left(\frac{x_j-u_\ell}{\epsilon}\right) \approx \frac{1}{\epsilon}\ln\left(\frac{x_j-u_\ell}{\epsilon}\right)$. Then replacing each into above we have
\begin{align*}
    \approx &\frac{x_j}{\epsilon}\left(\sum_{b\neq j} ^k    \ln  \left(\frac{x_j-x_b}{\epsilon} \right) -\sum_{c\neq j} ^k    \ln  \left(\frac{x_c-x_j}{\epsilon} \right)   -  \sum_{\ell=1}^n   \ln  \left(\frac{x_j-u_\ell}{\epsilon}\right)\right)+ \\
    &x_j \partial_{x_j} \ln \left(\frac{ z^{|\db|+(x_1+\cdots x_k)/\epsilon} }{\epsilon^{|\db| n}}\right).
\end{align*}

We have the first two sums cancel to $\ln (-1)^{k-1}$, and we can condense the third sum as $\ln (\prod (x_j-u_\ell))$. For the last term we have 
$$ x_j \partial_{x_j} \ln \left(\frac{ z^{|\db|+(x_1+\cdots x_k)/\epsilon} }{\epsilon^{|\db| n}}\right)=  x_j \partial_{x_j} \sum_{b=1}^k \ln \frac{z^{(x_b/\epsilon+d_b)}}{\epsilon^{nd_b}}= x_j \partial_{x_j} \ln \frac{z^{(x_j/\epsilon+d_j)}}{\epsilon^{nd_j}},$$
which is asymptotically equivalent to $x_j\partial_{x_j} \ln {z^{(x_j/\epsilon)}}$ as $\epsilon \to 0$. Since $\partial_{x_j} \ln {z^{(x_j/\epsilon)}}=\frac{1}{\epsilon}\ln z$, our expression above becomes
\begin{equation}
    \frac{x_j}{\epsilon}\left(\ln (-1)^{k-1} - \ln \left(\prod_{\ell=1} ^n (x_j-u_\ell)\right)+\ln z\right).
\end{equation}
Setting this equal to 0, eliminating $\frac{x_j}{\epsilon}$, and taking the exponential of both sides we obtain 
$$\prod_{\ell=1} ^n (x_j-u_\ell) = (-1)^{k-1} z,$$
completing the proof.
\end{proof}

The roots of this equation plugged into the factorial Schur polynomial give the eigenvalues of $M(z)^{C}$. For the eigenvectors of $M(z)^{C}$ we construct the off shell Bethe Vector as 
$$\begin{pmatrix}
    s_{\lambda^1} (\xb \, | \, \ub) \\
    s_{\lambda^2} (\xb \, | \, \ub)\\
    \vdots\\
    s_{\lambda^{\binom{n}{k}}}(\xb \, | \, \ub).
\end{pmatrix}.$$

Plugging in solutions of the Bethe equations for $\xb$ give the eigenvectors of $M(z)$.

\subsection{Nonequivariant quantum cohomology solutions}

The equivariant case degenerates to the nonequivariant case with the limits of $u_i \to 0$. We will use superscript $N$ to denote the nonequivariant versions of our objects. Here the qde is
\begin{equation} \label{psiCNEQ1n}
    \epsilon z \frac{\partial }{\partial z} \psi(z)^{N} = M(z)^{CN} \psi(z)^{N},
\end{equation}
where $M(z)^{CN}$ is $M(z)^{C}$ with limits $u_i\to 0$.

Taking limits $u_i \to 0$ of (\ref{psiCeq}), we can write the solution as the representative
\begin{equation*}
\psi_\lambda (\xb,z)^N = C \sum_{\db} \frac{\prod \limits_{b,c=1}^k \left(\frac{x_c-x_b}{\epsilon}+d_c-d_b+1\right) }{\prod \limits_{b=1} ^k  \Gamma \left(\frac{x_b}{\epsilon}+d_b+1\right)^n} s_\lambda (\xb+\db \epsilon) \left(\frac{z}{\epsilon^n}\right)^{|\db|}
\end{equation*}
where $C$ here is the same as the equivariant case (notice that $C$ did not depend upon the $u_i$'s).

\section{Quantum K-theory} \label{quantumktheorysection}

The main focus of this section is the quantum difference equation (QDE) for $\textrm{QK}_T(\textrm{Gr}(k,n))$. This is the equation
\begin{equation}
    \Psi(q z) \OC(-1) = M(z) \Psi(z),
\end{equation}
where $\Psi(z)$ is the solution we aim to find an explicit form for. Here, the $\OC(-1)$ denotes the regular action of tensoring by $\OC(-1)$ in equivariant K-theory. As Okounkov and Aganagic did in \cite{aganagic2017quasimapcountsbetheeigenfunctions}, we construct the solutions to the QDE using partition functions. First we define pieces of the QDE.

\subsection{The operators \texorpdfstring{$\OC(-1)$}{O(-1)} and \texorpdfstring{$M(z)$}{M(z)} }
Recall that $x_1, \ldots, x_k$ denote the (equivariant) K-theoretic Chern roots of the tautological subbundle. For $\textrm{K}_T(\textrm{Gr}(k,n))$,
$\OC(-1)$ has the representative
$$ 
\OC(-1)= \prod_{b=1} ^k x_b. 
$$
Let $M(z)$ denote the operator of quantum multiplication by $\OC(-1)$ in $\textrm{QK}_T(\textrm{Gr}(k,n))$. That is, for $V \in \mathrm{QK}_T (\mathrm{Gr}(k,n))$ $M(z) V = \mathcal{O} (-1) \star V$. Note that $M(0) =\OC(-1)$, i.e. $M(0)$ acts as the usual tensor product as in $\mathrm{K}_T (\mathrm{Gr}(k,n)).$

\begin{proposition}[From \cite{Buch_2011}]\label{kpieri}
For some $\lambda=(\lambda_1,\dots, \lambda_k)$ in the $k\times (n-k)$ rectangle, let $r=\{r_1,\dots,r_k\}$ be the corresponding $k$-subset. $M(z)$ acts on $\OC_\lambda$ as
\begin{equation}  
M(z) \mathcal{O}_{\lambda} = \left(\prod_{i=1} ^k {u_{r_i}} \right) \sum_{\nu} (-1)^{|\nu/\lambda|} \left(\mathcal{O}_{\nu}-z\mathcal{O}_{\hat{\overline{\nu}}} \right). 
\end{equation}
Here $\hat{\overline{\nu}}$ denotes the result of removing the first row and first column from $\nu$. The sum is over all partitions $\nu \supset \lambda$ in the $k\times (n-k)$ rectangle for which $\nu / \lambda$ is a rook strip. 
$|\nu / \lambda |$ is the number of boxes in the skew diagram.
\end{proposition}
 This formula is obtained from using $M(z)=\OC(-1)\star =(1-\OC_{1})\star$.
\begin{example}
 In $\textrm{QK}_T(\textrm{Gr}(2,4))$ we have the following:
\begin{align*}
    M(z)\OC_{(0,0)} &= \OC_{{0}}u_{{1}}u_{{2}}-\OC_{{1}}u_{{1}}u_{{2}} \\
    M(z)\OC_{(1,0)} &= \OC_{{1}}u_{{1}}u_{{3}}-\OC_{{2}}u_{{1}}u_{{3}}-\OC_{{1,1}}u_{{1}}u_{{3}}+\OC_
{{2,1}}u_{{1}}u_{{3}}
 \\
    M(z)\OC_{(1,1)} &= u_{{2}}u_{{3}}\OC_{{1,1}}-u_{{2}}u_{{3}}\OC_{{2,1}}\\
    M(z)\OC_{(2,0)} &= u_{{1}}u_{{4}}\OC_{{2}}-u_{{1}}u_{{4}}\OC_{{2,1}} \\
    M(z)\OC_{(2,1)} &= u_{{2}}u_{{4}}\OC_{{2,1}}-
\OC_{{2,2}}u_{{2}}u_{{4}}-z\OC_{{0}}u_{{2}}u_{{4}}+z\OC_{{1}}u_{{2}}u_{{4}}
 \\
    M(z)\OC_{(2,2)} &= \OC_{{2,2}}u_{{3}}u_{{4}} -z\OC_{{1}}u_{{3}}u_{{4}}.\\
\end{align*}
\end{example}


With these pieces defined, we derive the partition functions to construct the solution to the QDE.

\subsection{Partition functions for K-theory}
We replicate section \ref{cohpar}, passing the additive weights $u$ to multiplicative characters $t=e^{-u}$. For this K-theory case we use $1-t$ and $1-t^{-1}$ in place of $u$ and $-u$, respectively. Our diagram is:
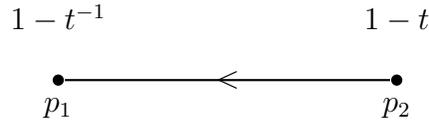
\[
\tikz[baseline=(p1.base), thick, >=stealth]{
  \node[circle, fill, inner sep=1.5pt, label=below:$p_1$] (p1) at (0,0) {};
  \node[circle, fill, inner sep=1.5pt, label=below:$p_2$] (p2) at (4.5,0) {};

  \draw[-] (p1) -- node[pos=0.5, yshift=-0.01ex] {${<}$} (p2);

  \node[above=4mm of p1] {$1 - t^{-1}$};
  \node[above=4mm of p2] {$1 - t$};
}
\]
Then we have attracting sets $Attr_{+} (p_1)=\bb{P}^1$ and $Attr_{+}(p_2)=p_2$. From this we have the positive attracting matrix defined with restrictions as $T_{ij}^+ = Attr_+ (j) |_{i}$:
\begin{equation}\label{kpplus}
    T^+ = \begin{pmatrix}
        Attr_+ (p_1) |_{p_1} &Attr_+ (p_2) |_{p_1}\\
        Attr_+ (p_1) |_{p_2}&Attr_+ (p_2) |_{p_2}
    \end{pmatrix}= \begin{pmatrix}
        1&0 \\ 1 &1-t
    \end{pmatrix}
\end{equation}

Then consider the same diagram with the opposite order:
\[
\tikz[baseline=(p1.base), thick, >=stealth]{
  \node[circle, fill, inner sep=1.5pt, label=below:$p_1$] (p1) at (0,0) {};
  \node[circle, fill, inner sep=1.5pt, label=below:$p_2$] (p2) at (4.5,0) {};

  \draw[-] (p1) -- node[pos=0.5, yshift=-0.01ex] {${>}$} (p2);

  \node[above=4mm of p1] {$1 - t^{-1}$};
  \node[above=4mm of p2] {$1 - t$};
}
\]
Then we have $Attr_{-} (p_1)=p_1$ and $Attr_{-}(p_2)=\bb{P}^1$. From this we have
\begin{equation}\label{kppneg}
    T^- = \begin{pmatrix}
        Attr_- (p_1) |_{p_1} &Attr_- (p_2) |_{p_1}\\
        Attr_- (p_1) |_{p_2}&Attr_- (p_2) |_{p_2}
    \end{pmatrix}= = \begin{pmatrix}
        1-t^{-1}&1 \\ 0 &1
    \end{pmatrix}.
\end{equation}
Then we compute using (\ref{kpplus}) and (\ref{kppneg}) to get
\begin{equation}
    (T^-)^{-1} \cdot T^+ = \begin{pmatrix}
        0&t\\
        1& 1-t
    \end{pmatrix}.
\end{equation}

From this we have our $R$ matrix

\begin{equation}\label{ktheoryrmatrix}
R(t)=\begin{pmatrix}
    1&0&0&0\\
    0&0&t&0\\
    0&1&1-t&0\\
    0&0&0&1
\end{pmatrix},
\end{equation}
which satisfies the Yang-Baxter equation
$$R_{a 1}(u_1/x) R_{a 2}(u_2/x) R_{12} (u_2/u_1)=R_{1 2}(u_2/u_1) R_{a 2}(u_2/x) R_{a 1}(u_1/x). $$
Partition functions from this $R$ matrix can be found using vertex models \cite{wheelerG},\cite{cotti2024satakecorrespondenceequivariantquantum}.

\begin{proposition}
The partition function in K-theory coincides with the double Grothendieck polynomials
\begin{equation} 
G_{\lambda} (\xb;\ub) = \frac{\det \left(x_j ^{i-1} (x_j \,|\, \ub) ^{\lambda_i +k -i} \right)_{1\leq i,j \leq k}}{\prod_{1\leq i<j\leq k }(x_j -x_i)},
\end{equation}
where $(x\, | \,\ub)$ is the {K-theoretic falling factorial}, 
defined as 
$$
(x \,|\, \ub)^b = (1-x/u_1)\cdots(1-x/u_b).
$$
\end{proposition}

These functions will be used to construct the solution to the QDE.

\subsection{\texorpdfstring{$\textrm{QK}_T(\textrm{Gr}(k,n))$}{QKT} case: solution to the QDE}

Denote $\xb q^{\boldsymbol{d}}=(x_1 q^{d_1},\dots, {x_k} q^{d_k})$. Taking a sum over $\db$ means to take the sum over $d_b\geq 0$, for $b\in [k]$. The $q-$Pochhammer symbol is defined as
$$
(x,q)_\infty = \prod_{a=1}^\infty (1-xq^a).
$$

\begin{theorem} \label{mainthm}
For any partition $\lambda$, the following class in $\mathrm{K}_T(\mathrm{Gr}(k,n))[[z]]$ solves the QDE $\Psi(qz) \OC(-1)= M(z) \Psi (z)$:
\begin{equation} \label{psiKEQkn}
    \Psi_\lambda(\xb,z)
    =\sum_{\db} \Phi(\xb q^{\boldsymbol{d}}, \ub)\, G_{\lambda} \left( \xb  q^{\boldsymbol{d}}; \ub\right)\, z^{|\boldsymbol{d}|}, 
\end{equation}
where  $\Phi$ is defined as
$$
\Phi(\xb,\ub)\coloneq\frac{\prod\limits_{b=1}^k \prod\limits_
{\ell=1}^n \left(\frac{x_b }{u_\ell}q,q\right)_\infty}{\prod\limits_{b,c=1}^k \left(\frac{x_b }{x_c }q,q\right)_\infty}.
$$
\end{theorem}

For the proof we first need two Lemmas. 
\begin{lemma} \label{pochsplit}
Let $$\mathcal{W}(\boldsymbol{d})= \prod_{1\leq i<j\leq k}  (x_j-x_i)^{-1} \left(\frac{-x_j}{x_i}\right)^{d_j-d_i} q^{{((d_j-d_i)(d_j-d_i-1)/2)}-d_i}.$$
Then we have 
$$\prod_{1\leq i<j\leq k} \left(x_j q^{d_j} - x_i q^{d_i}\right)=\mathcal{W}(\boldsymbol{d})^{-1}\prod_{b,c=1}^k \frac{\left(\frac{x_b}{x_c}q, q \right)_\infty}{\left(\frac{x_bq^{d_b}}{x_cq^{d_c}}q,q\right)_\infty}$$

\end{lemma}

\begin{proof}
Consider the product
$$\prod_{b,c=1}^k\frac{\left(\frac{x_b}{x_c}q, q \right)_\infty}{\left(\frac{x_bq^{d_b}}{x_cq^{d_c}}q,q\right)_\infty}.$$
The $b=c$ terms are 1, so we may write it as 
\begin{equation}\label{doubletrouble}
    \left(\prod_{b<c} \frac{\left(\frac{x_b}{x_c}q, q \right)_\infty}{\left(\frac{x_bq^{d_b}}{x_cq^{d_c}}q,q\right)_\infty}\right)\left(\prod_{b>c} \frac{\left(\frac{x_b}{x_c}q, q \right)_\infty}{\left(\frac{x_bq^{d_b}}{x_cq^{d_c}}q,q\right)_\infty}\right)=\left(\prod_{b<c} \frac{\left(\frac{x_b}{x_c}q, q \right)_\infty}{\left(\frac{x_bq^{d_b}}{x_cq^{d_c}}q,q\right)_\infty}\frac{\left(\frac{x_c}{x_b}q, q \right)_\infty}{\left(\frac{x_cq^{d_c}}{x_bq^{d_b}}q,q\right)_\infty}\right)
    \end{equation}

Note that for $\xi,\eta \in [k]$ where $\xi \neq \eta$, by using the $q$-Pochammer definition and simplifying further we obtain
    
    $$ \frac{\left(\frac{x_\xi}{x_\eta}q, q \right)_\infty}{\left(\frac{x_\xi q^{d_\xi}}{x_\eta q^{d_\eta}}q,q\right)_\infty}\times \frac{\left(\frac{x_\eta}{x_\xi}q, q \right)_\infty}{\left(\frac{x_\eta q^{d_\eta}}{x_\xi q^{d_\xi}}q,q\right)_\infty}= \frac{x_\xi q^{d_\xi}-x_\eta q^{d_\eta}}{x_\xi-x_{\eta}} \left(\frac{- x_\xi}{x_\eta}\right)^{d_\xi-d_\eta} q^{((d_\xi - d_\eta -1)(d_\xi - d_\eta)/2)-d_\eta}.$$
    
Then (\ref{doubletrouble}) becomes
\begin{align*}
    \left(\prod_{b<c} \frac{\left(\frac{x_b}{x_c}q, q \right)_\infty}{\left(\frac{x_bq^{d_b}}{x_cq^{d_c}}q,q\right)_\infty}\frac{\left(\frac{x_c}{x_b}q, q \right)_\infty}{\left(\frac{x_cq^{d_c}}{x_bq^{d_b}}q,q\right)_\infty}\right) &= \prod_{b<c} \frac{x_cq^{d_c}-x_b q^{d_b}}{x_c-x_{b}} \left(\frac{- x_c}{x_b}\right)^{d_c-d_b} q^{((d_c- d_b -1)(d_c- d_b)/2)-d_b}\\
    &= \left(\prod_{b<c} x_cq^{d_c}-x_b q^{d_b}\right) \mathcal{W}(\boldsymbol{d}),
\end{align*}
and after rearranging to isolate $\mathcal{W}(\boldsymbol{d})$ we are done.
\end{proof}

The second Lemma we use follows from Proposition \ref{pierigkprop} :
\begin{lemma}\label{gkpieri}
    Let $\lambda=(\lambda_1,\dots, \lambda_k)$ be a partition with $\lambda_1 = n-k$. For any partition $\nu=(\nu_1,\dots,\nu_k)$, define $\nu^a := (\nu_1+1,\nu_2,\dots,\nu_k)$. Then we have
    \begin{equation*}
        x_1\cdots x_k G_\lambda (\xb;\ub)  = u_{r_1}\cdots u_{r_k}\sum \limits_{\nu} (-1)^{|\lambda / \nu|} \left( G_\nu(\xb;\ub) -G_{\nu^a}(\xb;\ub)\right),
    \end{equation*}
    where this sum is over all partitions $\nu$ in the $k\times (n-k)$ rectangle such that $\nu \supset \lambda$ and $\nu/\lambda$ is a rook strip. $r$ is the corresponding $k$-subset to $\lambda$.
\end{lemma}
\begin{proof}
    First split the sum of Proposition \ref{pierigkprop} as
    \begin{equation*}
        x_1\cdots x_k G_\lambda(\xb;\ub)  = u_{r_1}\cdots u_{r_k}\left(\sum \limits_{\nu} (-1)^{|\lambda / \nu|} G_\nu (\xb;\ub)+\sum_\sigma (-1)^{|\sigma/\lambda|}G_{\sigma}(\xb ;\ub)\right),
    \end{equation*}
    where the first sum is over partitions $\nu$ that fit in $k\times (n-k)$ rectangle and the second sum is over partitions $\sigma=(\sigma_1,\dots,\sigma_k)$ that do not fit in this rectangle (all partitions still form rook strips when $\lambda$ is subtracted). However, note that for our case $\lambda_1=n-k$, for every partition $\nu\supset \lambda$ where $\nu/\lambda$ is a rook strip, $(\nu_1+1,\nu_2,\dots,\nu_k)$ is a partition (not fitting in the rectangle) that contains $\lambda$ and forms a rook strip when $\lambda$ is subtracted. Let $\nu^a=(\nu_1+1,\nu_2,\dots,\nu_k)$. Note we have each of these terms for every $\nu$, and for this case these are all the partitions outside the rectangle that form a rook strip when $\lambda$ is subtracted. Then our sum is
    \begin{align*}
        x_1\cdots x_k G_\lambda(\xb;\ub)  &= u_{r_1}\cdots u_{r_k}\left(\sum \limits_{\nu} (-1)^{|\lambda / \nu|} G_\nu(\xb;\ub) +\sum_\sigma (-1)^{|\sigma/\lambda|}G_{\sigma}(\xb;\ub)\right)\\
        &= u_{r_1}\cdots u_{r_k}\left(\sum \limits_{\nu} (-1)^{|\lambda / \nu|} G_\nu(\xb;\ub) +\sum_\nu (-1)^{|\nu^a/\lambda|}G_{\nu^a}(\xb;\ub)\right).
    \end{align*}
    Then since $(-1)^{|\nu^a/\lambda|}=-(-1)^{|\nu/\lambda|}$, we combine the sums and obtain the desired result.\end{proof}

With this information we now prove our main result.
\begin{proof}[Proof of Theorem \ref{mainthm}]
     The first case will consider the partitions $\lambda$ with $\lambda_1 \neq n-k$. The second case will cover partitions $\lambda$ with $\lambda_1 =n-k$. 
\vspace{0.5cm}

\noindent\emph{Case 1:}
Let $\lambda$ be a partition such that $\lambda_1 \neq n-k$. Observe

\begin{equation*}
    \Psi_\lambda(\xb, zq)  \OC(-1)=\sum  \limits_{\boldsymbol{d}} \Phi(\xb q^{\boldsymbol{d}},\ub) G_{\lambda} \left( \xb q^{\boldsymbol{d}}\right) (qz)^{|\boldsymbol{d}|} \cdot\left({x_1 \cdots  x_k}\right).
\end{equation*}
Since $q^{|\boldsymbol{d}|}=q^{d_1}\cdots q^{d_k}$, we can pair this to each of the $x_i$'s and the right hand side becomes
\begin{equation}\label{c1l2}
    \sum \limits_{\boldsymbol{d}} \Phi(\xb  q^{\boldsymbol{d}},\ub)\left({x_1q^{d_1}\cdots x_kq^{d_k}}\right)G_{\lambda} \left( \xb  q^{\boldsymbol{d}}\right)z^{|\boldsymbol{d}|}.
\end{equation}
Then using Lemma \ref{pierigkprop},  
$$
(x_1q^{d_1}\cdots x_k q^{d_k})G_\lambda (x_1q^{d_1}\cdots x_k q^{d_k}) = u_{r_1}\cdots u_{r_k} \sum_\mu  (-1)^{|\mu/\lambda|} G_\mu (x_1q^{d_1}\cdots x_k q^{d_k}).
$$
Note because $\lambda_1\neq n-k$, all the partitions $\mu$ indexing the sum are inside the $k\times (n-k)$ rectangle. We will use $\nu$ instead to indicate all the partitions $\nu$ are in the $k\times (n-k)$ rectangle. Then ($\ref{c1l2}$) becomes
\begin{equation} \label{rooky}
    \sum \limits_{\boldsymbol{d}} \Phi(\xb  q^{\boldsymbol{d}},\ub)\left(\prod_{b=1} ^k {u_{r _b}} \right) \left(\sum_{\nu} (-1)^{|\nu/\lambda|} G_{\nu}(\xb q^{\boldsymbol{d}})\right)z^{|\boldsymbol{d}|}.
\end{equation}
By Proposition \ref{kpieri}, we write (\ref{rooky}) as
\begin{equation*}
    \sum \limits_{\boldsymbol{d}} \Phi(\xb q^{\boldsymbol{d}},\ub)\left(M(z)G_\lambda\left(\xb q^{\boldsymbol{d}}\right)\right)z^{|\boldsymbol{d}|}.
\end{equation*}
Since $M(z)$ does not depend on $\boldsymbol{d}$, this is
\begin{align*}
    &=M(z)\sum \limits_{\boldsymbol{d}} \Phi(\xb q^{\boldsymbol{d}},\ub)G_\lambda\left(\xb q^{\boldsymbol{d}}\right)z^{|\boldsymbol{d}|}\\
    &=M(z)\Psi_\lambda (z).
\end{align*}
\vspace{0.2cm}

\noindent\emph{Case 2:} Choose $\lambda$ such that $\lambda_1 = n-k$. We show that $\Psi(qz) \OC(-1) - M(z) \Psi_\lambda(z)=0$. 

Observe for $\Psi_\lambda(qz)  \OC(-1)$ we may follow the first few steps of case 1 to obtain
\begin{equation}\label{c2l1}
    \Psi_\lambda(qz) \OC(-1) \!- \! M(z)\Psi_\lambda(z)
     \!=\! \sum \limits_{\boldsymbol{d}} \!\Phi(\xb q^{\boldsymbol{d}},\ub) \left( x_1q^{d_1}\cdots x_kq^{d_k}\right)G_{\lambda} \left( \xb q^{\boldsymbol{d}}\right)z^{|\boldsymbol{d}|}  -M(z)\Psi_\lambda(z).
\end{equation}
We consider each piece separately and then combine. For the first part of (\ref{c2l1}), since $\lambda_1=n-k$ we use Lemma \ref{gkpieri} to have that 
\begin{equation} \label{gkppf}
x_1q^{d_1}\cdots x_kq^{d_k} G_\lambda(\xb q^{\db};\ub)  = u_{r_1}\cdots u_{r_k}\sum \limits_{\nu} (-1)^{|\lambda / \nu|} \left( G_\nu(\xb q^{\db};\ub)  -G_{\nu^a}(\xb q^{\db};\ub) \right),
\end{equation}
where this sum is over partitions $\nu$ in the $k\times (n-k)$ rectangle such that $\nu/\lambda$ is a rook strip. Recall $\nu^a$ is defined by adding 1 to the first entry of $\nu$ (i.e. $\nu^a=(\nu_1+1,\nu_2,\dots, \nu_k))$.

For the second part of (\ref{c2l1}), we obtain
\begin{align*}
    M(z)\Psi_\lambda(z) &= M(z)\sum \limits_{\boldsymbol{d}} \Phi(\xb q^{\boldsymbol{d}},\ub)G_\lambda\left(\xb q^{\boldsymbol{d}}\right)z^{|\boldsymbol{d}|}\\
    &=\sum \limits_{\boldsymbol{d}} \Phi(\xb q^{\boldsymbol{d}},\ub)M(z)G_\lambda\left(\xb q^{\boldsymbol{d}}\right)z^{|\boldsymbol{d}|}.
\end{align*}
Recall that $\hat{\overline{\nu}}$ is obtained by removing the first row and column from $\nu$.
By Proposition \ref{kpieri}, we have $M(z)G_\lambda= u_{r_1}\cdots u_{r_k} \sum_\nu (-1)^{|\nu/\lambda|} (G_\nu (\xb q^{\boldsymbol{d}})-zG_{\hat{\overline{\nu}}}(\xb q^{\boldsymbol{d}})).$ Using this and (\ref{gkppf}), we write equation (\ref{c2l1}) as: 
\begin{align*}
    \Psi_\lambda(qz) \OC(-1) - M(z)\Psi_\lambda(z)
     =& \sum \limits_{\boldsymbol{d}} z^{|\boldsymbol{d}|}\Phi(\xb q^{\boldsymbol{d}},\ub) \left( x_1q^{d_1}\cdots x_kq^{d_k}\right)G_{\lambda} \left( \xb q^{\boldsymbol{d}}\right) -M(z)\Psi_\lambda(z)\\
     = \sum \limits_{\boldsymbol{d}}  z^{|\boldsymbol{d}|}\Phi(\xb q^{\boldsymbol{d}},\ub) &\bigg({u_{r_1}\cdots u_{r_k}}\sum \limits_\nu (-1)^{|\nu/\lambda|} (G_{\nu}(\xb q^{\boldsymbol{d}})- G_{\nu^a}(\xb q^{\boldsymbol{d}}) ) \\
     &     -{u_{r_1}\cdots u_{r_k}} \sum \limits_\nu (-1)^{|\nu/\lambda|} (G_\nu(\xb q^{\boldsymbol{d}}) - z G_{\hat{\overline{\nu}}}(\xb q^{\boldsymbol{d}}))\bigg)\\
    ={u_{r_1}\cdots u_{r_k}} \sum_\nu (-1)^{|\nu/\lambda|}  &\left(\sum_{\boldsymbol{d}} z^{|\boldsymbol{d}|}\Phi(\xb q^{\boldsymbol{d}},\ub)\left(-G_{\nu^a}(\xb q^{\boldsymbol{d}}) +z G_{\hat{\overline{\nu}}}(\xb q^{\boldsymbol{d}})\right)\right)
\end{align*}
We will show that the $ \sum_{\boldsymbol{d}} z^{|\boldsymbol{d}|}\Phi(\xb q^{\boldsymbol{d}},\ub)\left(-G_{\nu^a}(\xb q^{\boldsymbol{d}}) +z G_{\hat{\overline{\nu}}}(\xb q^{\boldsymbol{d}})\right)$ part of the sum to telescopes to 0 for each $\nu$. Choose a partition $\nu$. For convenience let $y_b=x_bq^{d_b}$ for each $b\in [k]$. Then using the determinant formula (\ref{gdetformula}) to unravel 
$$
\sum_{\boldsymbol{d}} z^{|\boldsymbol{d}|}\Phi(\xb q^{\boldsymbol{d}},\ub)\left(-G_{\nu^a}(\xb q^{\boldsymbol{d}}) +z G_{\hat{\overline{\nu}}}(\xb q^{\boldsymbol{d}})\right),
$$
we have:
\begin{equation} \label{detfirst}
\begin{split}
    &\sum_{\boldsymbol{d}} z^{|\boldsymbol{d}|}\Phi (\boldsymbol{y},\ub) \left(-G_{\nu^a}(\boldsymbol{y};\ub) +z G_{\hat{\overline{\nu}}}\right(\boldsymbol{y};\ub))=\\
    &\scalebox{0.93}{$ \displaystyle -\sum_{\boldsymbol{d}} z^{|\boldsymbol{d}|}\Phi \,V^{-1} \! \left(\begin{vmatrix}
(y_1\!\mid\!\ub)^{\nu_1 + k} &  \! \cdots \!& (y_k\!\mid\!\ub)^{\nu_1 + k} \\[6pt]
y_1\,(y_1\!\mid\!\ub)^{\nu_2 + k - 2} &   \! \cdots \!& y_k\,(y_k\!\mid\!\ub)^{\nu_2 + k - 2} \\
\vdots &  \! \ddots \! & \vdots \\
y_1^{\,k-1}\,(y_1\!\mid\!\ub)^{\nu_k} &   \! \cdots \!& y_k^{\,k-1}\,(y_k\!\mid\!\ub)^{\nu_k}
\end{vmatrix}\!- \! z\begin{vmatrix}
(y_1\!\mid\!\ub)^{\nu_1 + k - 2} &   \! \cdots \! & (y_k\!\mid\!\ub)^{\nu_1 + k - 2} \\
\vdots  &  \! \ddots \! &\vdots \\
y_1^{\,k-2}\,(y_1\!\mid\!\ub)^{\nu_k} &   \! \cdots \!& y_k^{\,k-2}\,(y_k\!\mid\!\ub)^{\nu_k} \\[6pt]
y_1^{\,k-1} &  \! \cdots \! & y_k^{\,k-1}
\end{vmatrix}\right)$},
\end{split}
\end{equation} 
where $V= \prod_{i<j} (y_j-y_i)$ is the Vandermonde determinant. We will rewrite each of these determinants. For the first matrix, note by factoring $y_j$ from each column we obtain
\begin{equation*}
    \begin{vmatrix}
(y_1\!\mid\!\ub)^{\nu_1 + k} &  \cdots & (y_k\!\mid\!\ub)^{\nu_1 + k} \\[6pt]
y_1\,(y_1\!\mid\!\ub)^{\nu_2 + k - 2} &  \cdots & y_k\,(y_k\!\mid\!\ub)^{\nu_2 + k - 2} \\
\vdots & \ddots & \vdots \\
y_1^{\,k-1}\,(y_1\!\mid\!\ub)^{\nu_k} &  \cdots & y_k^{\,k-1}\,(y_k\!\mid\!\ub)^{\nu_k}
\end{vmatrix}=y_1\cdots y_k \begin{vmatrix}
y_1^{-1}(y_1\!\mid\!\ub)^{\nu_1 + k} &  \cdots & y_k^{-1}(y_k\!\mid\!\ub)^{\nu_1 + k} \\[6pt]
\,(y_1\!\mid\!\ub)^{\nu_2 + k - 2} &  \cdots & \,(y_k\!\mid\!\ub)^{\nu_2 + k - 2} \\
\vdots & \ddots & \vdots \\
y_1^{\,k-2}\,(y_1\!\mid\!\ub)^{\nu_k} &  \cdots & y_k^{\,k-2}\,(y_k\!\mid\!\ub)^{\nu_k} 
\end{vmatrix}.
\end{equation*}
Note that $\nu_1=n-k$ so $\nu_1+k=n$. Rewriting this and then multiplying the factor in front to the first row turns this factor into
\begin{equation} \label{detl2}
\begin{vmatrix}
(y_1\!\mid\!\ub)^{n}\prod_{m\in [k]\setminus {1}} y_m  &  \cdots &(y_k\!\mid\!\ub)^{n}  \prod_{m\in [k]\setminus {k}}y_m \\[6pt]
\,(y_1\!\mid\!\ub)^{\nu_2 + k - 2} &  \cdots & \,(y_k\!\mid\!\ub)^{\nu_2 + k - 2} \\
\vdots & \ddots & \vdots \\
y_1^{\,k-2}\,(y_1\!\mid\!\ub)^{\nu_k} &  \cdots & y_k^{\,k-2}\,(y_k\!\mid\!\ub)^{\nu_k} 
\end{vmatrix}
\end{equation}
Now for the second term of (\ref{detfirst}) we switch the bottom row with the row above it $k-1$ times so that it is at the top, multiplying $(-1)$ with each switch:
\begin{equation}\label{detl3}
    z\begin{vmatrix}
(y_1\!\mid\!\ub)^{\nu_1 + k - 2} &  \cdots & (y_k\!\mid\!\ub)^{\nu_1 + k - 2} \\
\vdots  & \ddots &\vdots\\
y_1^{\,k-2}\,(y_1\!\mid\!\ub)^{\nu_k} &  \cdots & y_k^{\,k-2}\,(y_k\!\mid\!\ub)^{\nu_k} \\[6pt]
y_1^{\,k-1} &  \cdots & y_k^{\,k-1}
\end{vmatrix}=(-1)^{k-1}z\begin{vmatrix}
y_1^{\,k-1} &  \cdots & y_k^{\,k-1}\\
(y_1\!\mid\!\ub)^{\nu_1 + k - 2} &  \cdots & (y_k\!\mid\!\ub)^{\nu_1 + k - 2} \\
\vdots  & \ddots &\vdots\\
y_1^{\,k-2}\,(y_1\!\mid\!\ub)^{\nu_k} &  \cdots & y_k^{\,k-2}\,(y_k\!\mid\!\ub)^{\nu_k} 
\end{vmatrix}
\end{equation}
Attaching this negative factor to the first row, we use (\ref{detl2}) and (\ref{detl3}) to write (\ref{detfirst}) as:
\begin{equation}\label{combinedps}
    -\sum_{\boldsymbol{d}} \, z^{|\db|}\Phi(\boldsymbol{y},\ub) V^{-1}\begin{vmatrix}
    (y_j \, | \, \ub)^n \prod\limits_ {{m}\neq j}^k y_{m} - z (-y_j)^{k-1}\\
    (y_j \, | \, \ub)^{\nu_2+k-2}\\
    \vdots\\
    y_j ^{k-2} (y_j \, | \, \ub)^{\nu_k}
\end{vmatrix}_{j=1} ^k.
\end{equation}
We omit writing the negative on the outside of the sum. Expanding this determinant into minors across the first row will give us $k$ minors. We will show that each of these become a sum which telescopes to 0.
\vspace{0.4cm}

\noindent
Pick $\eta\in [k]$. Let
\begin{equation}
    \mathcal{M}_\eta = \begin{vmatrix}
(y_1\!\mid\!\ub)^{\nu_1 + k - 2} & (y_2\!\mid\!\ub)^{\nu_1 + k - 2} & \cdots &\widehat{ (y_\eta\!\mid\!\ub)}^{\nu_1 + k - 2}&\cdots&(y_k\!\mid\!\ub)^{\nu_1 + k - 2} \\
\vdots & \vdots & & \vdots \\
y_1^{\,k-2}\,(y_1\!\mid\!\ub)^{\nu_k} & y_2^{\,k-2}\,(y_2\!\mid\!\ub)^{\nu_k} & \cdots &\widehat{y_\eta ^{k-2} (y_\eta\!\mid\!\ub)}^{\nu_k}&\cdots& y_k^{\,k-2}\,(y_k\!\mid\!\ub)^{\nu_k} \\
\end{vmatrix},
\end{equation}
the $(k-1)\times (k-1)$ minor to the top row $\eta$ term. With this we can write the $\eta$ term of the expansion (\ref{combinedps}) as 
\begin{equation} \label{VMsum} 
\sum_{\boldsymbol{d}} z^{|\db|}\Phi(\boldsymbol{y},\ub) \,(-1)^{\eta+1}V^{-1} \mathcal{M}_\eta \left[  (y_\eta \, | \, \ub)^{\nu_1+k} \prod _{m\neq \eta} ^k y_m - z\, (-y_\eta) ^{k-1}\right].
\end{equation}

We can factor out the $(-1)^{\eta+1}$ term and we will omit writing it. We will split this into two sums:
\begin{align*}
    &\sum_{\boldsymbol{d}}  V ^{-1}\mathcal{M}_\eta \, \Phi(\xb q^{\boldsymbol{d}},\ub){z}^{|\boldsymbol{d}|}  \cdot \left(  (y_\eta \, | \, \ub)^{n} \prod _{m\neq \eta} ^k y_m\right) \\&\hspace{2.5in}
- \sum_{d_1,\dots,d_k \geq 0} V ^{-1}\mathcal{M}_\eta \, \Phi(\xb q^{\boldsymbol{d}},\ub){z}^{|\boldsymbol{d}|}\left(z\, (-y_\eta )^{k-1}\right).
\end{align*}
Note that when $d_\eta=0$, $(y_\eta \, | \, \ub)^{n} =(1-x_\eta /u_1)\cdots (1-x_\eta /u_n)$, meaning restricting $x_\eta$ to any fixed point would make $(y_\eta \, | \, \ub)^n=0$. From equivariant localization of this ring, the $d_\eta=0$ terms of the first sum are 0. Hence we have the above equivalent to:

\begin{align}
&\sum_{\substack{d_m\geq 0 \text{ for } m\neq\eta \\d_\eta \geq 1}}  V ^{-1}\mathcal{M}_\eta \, \Phi(\xb q^{\boldsymbol{d}},\ub){z}^{|\boldsymbol{d}|}  \cdot \left(  (y_\eta \, | \, \ub)^{n} \prod _{m\neq \eta} ^k y_m\right)\nonumber \\&\hspace{2.5in}
- \sum_{d_1,\dots,d_k \geq 0} V ^{-1}\mathcal{M}_\eta \, \Phi(\xb q^{\boldsymbol{d}},\ub){z}^{|\boldsymbol{d}|}\left(z\, (-y_\eta )^{k-1}\right)\label{splitlesseta}
\end{align}

We expand out $\Phi$ and use Lemma \ref{pochsplit} to simplify further. Recall that the Lemma states $$V=\mathcal{W}(\boldsymbol{d})^{-1}\prod_{b,c=1}^k \frac{\left(\frac{x_b}{x_c}q, q \right)_\infty}{\left(\frac{x_bq^{d_b}}{x_cq^{d_c}}q,q\right)_\infty}.$$\\
  Then (\ref{splitlesseta}) is 
\begin{align*}
=&\sum_{\substack{d_m\geq 0 \text{ for } m\neq\eta \\d_\eta \geq 1}} \mathcal{M}_\eta  \mathcal{W}(\boldsymbol{d})\prod_{b,c=1}^k \frac{\left(\frac{x_bq^{d_b}}{x_cq^{d_c}}q,q\right)_\infty}{\left(\frac{x_b}{x_c}q, q \right)_\infty}\cdot \frac{\prod_{b=1} ^k \prod_{l=1}^n \left(\frac{x_b q^{d_b}}{u_\ell}q,q \right)_\infty}{\prod\limits_{b,c=1}^k \left(\frac{x_b \,q^{d_b}}{x_c \,q^{d_c}}q,q\right)_\infty}{z}^{|\boldsymbol{d}|}  \left(  (y_\eta \, | \, \ub)^{n} \prod _{m\neq \eta} ^k y_m\right) \\
&
- \sum_{d_1,\dots,d_k \geq 0} \mathcal{M}_\eta \left(\mathcal{W}(\boldsymbol{d})\prod_{b,c=1}^k \frac{\left(\frac{x_bq^{d_b}}{x_cq^{d_c}}q,q\right)_\infty}{\left(\frac{x_b}{x_c}q, q \right)_\infty}\right) \frac{\prod_{b=1} ^k \prod_{l=1}^n \left(\frac{x_b q^{d_b}}{u_\ell}q,q \right)_\infty}{\prod\limits_{b,c=1}^k \left(\frac{x_b \,q^{d_b}}{x_c \,q^{d_c}}q,q\right)_\infty}{z}^{|\boldsymbol{d}|}\left(z\, (-y_\eta )^{k-1}\right).
\end{align*}
Simplify further and factor out the $\prod_{b,c=1} ^k \left( \frac{x_b}{x_c}q, q \right)^{-1}_\infty$ term from both sums. We will omit writing this factor and show the rest becomes 0.
\begin{align}\label{nearfinish}
=&\sum_{\substack{d_m\geq 0 \text{ for } m\neq\eta \\d_\eta \geq 1}} \mathcal{M}_\eta  \, \mathcal{W}(\boldsymbol{d})\cdot {\prod_{b=1} ^k \prod_{l=1}^n \left(\frac{x_b q^{d_b}}{u_\ell}q,q \right)_\infty}{z}^{|\boldsymbol{d}|}  \cdot \left(  (y_\eta \, | \, \ub)^{n} \prod _{m\neq \eta} ^k y_m\right) \nonumber \\
&- \sum_{d_1,\dots,d_k \geq 0} \mathcal{M}_\eta\, \mathcal{W}(\boldsymbol{d})\prod_{b,c=1}^k  {\prod_{b=1} ^k \prod_{l=1}^n \left(\frac{x_b q^{d_b}}{u_\ell}q,q \right)_\infty}{z}^{|\boldsymbol{d}|}\left(z\, (-y_\eta )^{k-1}\right).
\end{align}
Note that by algebraic manipulation we obtain
\begin{align*}
\prod_{b=1} ^k \prod_{l=1}^n \left(\frac{x_b q^{d_b}}{u_\ell}q,q \right)_\infty(y_\eta \, | \, \ub)^{n}&=\prod_{b=1} ^k \prod_{l=1}^n \left(\frac{x_b q^{d_b}}{u_\ell}q,q \right)_\infty \cdot \left(1-\frac{x_\eta q^{d_\eta}}{u_1}\right)\cdots\left(1-\frac{x_\eta q^{d_\eta}}{u_n}\right)\\
&= \prod_{b\in [k]\setminus  \{\eta\}}  \prod_{l=1}^n \left(\frac{x_b q^{d_b}}{u_\ell}q,q \right)_\infty \prod_{l=1}^n \left(\frac{x_\eta q^{d_\eta}}{u_\ell},q \right)_\infty.
\end{align*}
We also combine $\prod_{m\neq \eta} ^k y_m$ with $\mathcal{W}(\boldsymbol{d})$ to get
$$\mathcal{W}(\boldsymbol{d})\prod_{m\neq \eta} ^k y_m = \mathcal{W}(d_1\dots, d_{\eta-1},d_\eta-1,d_{\eta+1},\dots,d_k) \, (-x_\eta q^{d_\eta})^{k-1}.
$$
Combining these facts with (\ref{nearfinish}) changes the first sum of the expression and we have

\begin{align*}
=&\sum_{\substack{d_m\geq 0 \text{ for } m\neq\eta \\d_\eta \geq 1}} \mathcal{M}_\eta \, \mathcal{W}(d_{\eta}-1)\, {\prod_{b\in [k]\setminus  \{\eta\}}  \prod_{l=1}^n \left(\frac{x_b q^{d_b}}{u_\ell}q,q \right)_\infty \prod_{l=1}^n \left(\frac{x_\eta q^{d_\eta}}{u_\ell},q \right)_\infty}{z}^{|\boldsymbol{d}|}  \cdot \left( -x_\eta q^{d_\eta-1} \right) ^{k-1}\\
&- \sum_{d_1,\dots,d_k \geq 0} \mathcal{M}_\eta\mathcal{W}(\boldsymbol{d})\, {\prod_{b=1} ^k \prod_{l=1}^n \left(\frac{x_b q^{d_b}}{u_\ell}q,q \right)_\infty}{z}^{|\boldsymbol{d}|+1}\, (-x_\eta q^{d_\eta})^{k-1}
\end{align*}

Define a function $f_\eta (d_\eta) =\mathcal{M}_\eta\mathcal{W}(\boldsymbol{d})\, {\prod_{b=1} ^k \prod_{l=1}^n \left(\frac{x_b q^{d_b}}{u_\ell}q,q \right)_\infty}{z}^{|\boldsymbol{d}|+1}\, (-x_\eta q^{d_\eta})^{k-1}$. We write it only in terms of $d_\eta$ since that will be the only variable shifted. Then the above expression is
\begin{equation*}
    \sum_{\substack{d_m\geq 0 \text{ for } m\neq\eta \\d_\eta \geq 1}} f_\eta (d_\eta-1) - \sum_{d_1,\dots,d_k \geq 0}  f_\eta(d_\eta).
\end{equation*}
Shift the first sum $d_\eta \to d_\eta+1$. After reindexing we obtain
\begin{equation*}
    \sum_{d_1,\dots,d_k \geq 0}  f_\eta(d_\eta)-\sum_{d_1,\dots,d_k \geq 0}  f_\eta(d_\eta)=0,
\end{equation*}
showing that the minor of the $\eta$ term when expanding across the first row is 0. Since $\eta$ was picked arbitrarily, this proves the other $k-1$ terms become 0 as well, completing the proof.
\end{proof}

\subsection{Bethe equations via saddle point and Bethe vectors}

\begin{proposition}
    At $q\to 1$ the saddle point of the integral (\ref{psiKEQkn}) corresponds to the $k$ Bethe equations
    \begin{equation}
        z \prod \limits_{i\in [k]\setminus\  \{j\}}\frac{-x_j}{x_i} = \prod \limits_{\ell=1}^n \left(1-\frac{x_j}{u_\ell}\right), 
    \end{equation}
    for $j\in [k]$.
\end{proposition}
\begin{proof}
   In the solution (\ref{psiKEQkn}), note $z^{d_i}$ can be written as
$$z^{d_i} = z^{\frac{\ln x_i q^{d_i}}{\ln q}}z^{\frac{-\ln x_i}{\ln q}}.$$
   Rewrite each $z^{d_i}$ of $z^{|\db|}$ as this.
   Let $\varphi$ denote the inside of the sum except for the double Grothendieck polynomials and each $z^{\frac{-\ln x_i}{\ln q}}$ (keep the numerator). The saddle point is defined by the equations $x_j \partial_j \ln(\varphi) =0$ for $j\in [k]$. First we use the following Lemma:

\begin{lemma}
    Take limit $q\to 1$. Asymptotically, we have
   $$x\frac{\partial \ln ((x,q)_\infty)}{\partial x} = \frac{-\ln(1-x)}{\ln(q)} +o(\ln(q)). $$
\end{lemma}
\begin{proof}
    Observe that 
    \begin{align*}
\ln (x,q)_\infty
&= \ln \prod_{m=0}^\infty (1 - x q^m) \\
&= \sum_{m=0}^\infty \ln(1 - x q^m) \\
&= - \sum_{m=0}^\infty \sum_{b=1}^\infty \frac{(x q^m)^b}{b} \\
&= - \sum_{b=1}^\infty \frac{x^b}{b} \sum_{m=0}^\infty q^{bm} \\
&= - \sum_{b=1}^\infty \frac{x^b}{b(1 - q^b)} \\
&\approx - \sum_{b=1}^\infty \frac{x^b}{-\,b^2 \ln q} \\
&= \frac{1}{\ln q} \sum_{b=1}^\infty \frac{x^b}{b^2} \\
&= \frac{\mathrm{Li}_2(x)}{\ln q},
\end{align*}
where the approximation $q^b=e^{b\ln q} \approx 1 +b\ln q$ was used since $q\to 1$. By taking derivatives we obtain our desired result.
\end{proof}
Also note that as $q\to 1$, we have $x_iq^{d_i} \approx x_i$. Then observe for $j\in [k]:$ 

\begin{align*}
x_j\frac{\partial}{\partial x_j}\ln(\varphi) &= 0\\
x_j\frac{\partial}{\partial x_j}\ln\left[\prod \limits_{b,c=1} ^k(x_j/x_c,q)_\infty\;\prod_{b=1}^k\prod_{\ell=1}^n(x_b/u_\ell ,q)_\infty^{-1}\; \prod \limits_{c=1} ^k z^{\frac{\ln x_c}{\ln q}}\right]
&=0
\\
\frac{1}{\ln q}\left(\sum_{c\neq j} -\ln\Bigl(1 - \tfrac{x_j}{x_c}\Bigr)
\;+\;
\sum_{b\neq j} \ln\Bigl(1 - \tfrac{x_b}{x_j}\Bigr)
\;+\;
\sum_{\ell=1}^n \ln\Bigl(1 - \tfrac{x_j}{u_\ell}\Bigr)
\;-\;\ln z
\right)&=0\\
\sum_{c\neq j} -\ln\Bigl(1 - \tfrac{x_j}{x_c}\Bigr)
\;+\;
\sum_{b\neq j} \ln\Bigl(1 - \tfrac{x_b}{x_j}\Bigr)
\;+\;
\sum_{\ell=1}^n \ln\Bigl(1 - \tfrac{x_j}{u_\ell}\Bigr)
\;-\;\ln z
&=0.
\end{align*}

Then by exponentiating both sides we obtain:
\begin{align*}
\exp\Biggl[
  -\sum_{c\neq j}\ln\Bigl(1 - \tfrac{x_j}{x_c}\Bigr)
  + \sum_{b\neq j}^k\ln\Bigl(1 - \tfrac{x_b}{x_j}\Bigr)
  + \sum_{\ell=1}^n\ln\Bigl(1 - \tfrac{x_j}{u_\ell}\Bigr)
\Biggr]
&=z
\\
\frac{
  \displaystyle\prod_{b\neq j}\Bigl(1 - \tfrac{x_b}{x_j}\Bigr)
  \;\;\displaystyle\prod_{\ell=1}^n\Bigl(1 - \tfrac{x_j}{u_\ell}\Bigr)
}{
  \displaystyle\prod_{c\neq j}\Bigl(1 - \tfrac{x_j}{x_c}\Bigr)
}
&=z
\\[6pt]
\prod \limits_{l=1}^n \left(1-\frac{x_j}{u_\ell}\right)&= z \prod \limits_{b\in \{1..k\} \setminus\  j}\frac{-x_j}{x_b} .
\end{align*}
\end{proof}

The solutions to the Bethe equations give eigenvalues of $M(z)$. Since the representative for $M(z)$ is $1-G_1(\xb ;\ub)=x_1\cdots x_k/{u_1\cdots u_k}$, the eigenvalues are formed by products of solutions to the Bethe equations.

The off shell Bethe Vector is
$$\begin{pmatrix}
    G_{\lambda^1} (\xb ;\ub) \\
    G_{\lambda^2} (\xb ;\ub)\\
    \vdots\\
    G_{\lambda^{\binom{n}{k}}}(\xb ;\ub).
\end{pmatrix}.$$

For each Bethe equation, substituting its roots for $\xb $ gives an eigenvector $M(z)$.

\subsection{Solution to QDE in \texorpdfstring{$\textrm{QK}(\textrm{Gr}(k,n))$}{QK}}

For the nonequivariant case, the QDE appears as 
\begin{equation}
    \Psi(qz)^{N} M(0) = M(z) \Psi(z)^{N}.
\end{equation}

Take the equivariant solution, and take the limit of the $u_\ell$ parameters to 1. We obtain:

\begin{equation}
     \Psi^N_\lambda(\xb ,z)
    =\sum_{\db} \Phi(\xb q^{\boldsymbol{d}})\, G_{\lambda} \left( \xb q^{\boldsymbol{d}}\right)\, z^{|\boldsymbol{d}|}, 
\end{equation}
where  $\Phi$ is defined as

$$\Phi(\xb ,\ub)\coloneq\frac{\prod\limits_{b=1}^k \prod\limits_
    {\ell=1}^n \left({x_b \, }q,q\right)_\infty}{\prod\limits_{b,c=1}^k \left(\frac{x_b }{x_c }q,q\right)_\infty} .$$

$G_\lambda$ is the stable Grothendieck polynomial   corresponding to $\lambda$ and\\ $G_\lambda(\xb q^{\boldsymbol{d}}  )=G_\lambda (x_1q^{d_1},\dots, x_kq^{d_k} )$.

\printbibliography

\end{document}